\documentclass[a4paper,fleqn]{cas-dc}

\usepackage[numbers]{natbib}

\def\tsc#1{\csdef{#1}{\textsc{\lowercase{#1}}\xspace}}
\tsc{WGM}
\tsc{QE}
\tsc{EP}
\tsc{PMS}
\tsc{BEC}
\tsc{DE}
\usepackage{graphicx}
\usepackage{subcaption}
\usepackage{etoolbox}
\usepackage{scalerel}

\usepackage{hyperref}
\usepackage{siunitx}
\usepackage{pifont}
\usepackage{graphicx}
\usepackage{textcomp}
\usepackage{algorithm} 
\usepackage{cite}
\usepackage{booktabs}
\usepackage{url}
\usepackage{algpseudocode}
\usepackage{paralist}
\usepackage{color}
\usepackage{epsfig} % for postscript graphics files
\usepackage{bm} % assumes amsmath package installed
\usepackage{listings,lstautogobble}
\usepackage{color}
\usepackage{xspace}
\usepackage{soul} %,xcolor
\usepackage{verbatim}
\usepackage{epstopdf}
\usepackage{listings}
\usepackage{parcolumns}
\usepackage{color}
\usepackage{mathtools}
\usepackage{xspace}
\usepackage{verbatim}
\usepackage[utf8]{inputenc}
\usepackage{rotating,multirow}

\usepackage[most]{tcolorbox}
\usepackage{amsthm}
\newtheorem{definition}{Definition}
\newtheorem{proposition}{Proposition}
\newtheorem{theorem}{Theorem}
\newtheorem{lemma}{Lemma} 
  \renewcommand{\footnotesize}{\fontsize{8pt}{11pt}\selectfont}

\newcommand{\cmark}{\ding{51}}%
\newcommand{\xmark}{\ding{55}}%
\newcommand{\sm}[1]{\scaleto{#1}{3.5pt}}

\DeclareMathAlphabet{\mathcal}{OMS}{cmsy}{m}{n}

\DeclareMathOperator*{\argminB}{argmin}   % Jan Hlavacek

\newcommand{\zono}[1]{\langle #1 \rangle}

\newcommand{\R}{{\mathbb{R}}}

\begin{document}

\let\WriteBookmarks\relax
\def\floatpagepagefraction{1}
\def\textpagefraction{.001}
\shorttitle{Set-Based Diffusion}
\shortauthors{Amr Alanwar et~al.}

\title[mode = title]{Distributed Set-Based Observers Using Diffusion Strategies}
\author[1,2]{Amr Alanwar}[orcid=0000-0003-2941-519X]
%\cormark[1]
%\fnmark[1]
\ead{alanwar@kth.se}
\author[3]{Jagat Jyoti Rath}[orcid=0000-0002-8365-3538]
%\cormark[1]
%\fnmark[1]
\ead{jagatjyoti.rath@gmail.com}
\author[4]{Hazem Said}[orcid=0000-0003-2419-4375]
%\cormark[1]
\ead{hazem.said@eng.asu.edu.eg}
\author[1]{Karl Henrik Johansson}[orcid=0000-0001-9940-5929]
%\cormark[1]
\ead{kallej@kth.se}
\author[5]{Matthias Althoff}[orcid=0000-0003-3733-842X]
%\cormark[1]
\ead{althoff@tum.de}

\address[1]{School of Electrical Engineering and Computer Science, KTH Royal Institute of Technology, Sweden}
\address[2]{School of Computer Science and Engineering, Constructor University, Germany}
\address[3]{Department of Mechanical and Aero-Space Engineering, Institute of Infrastructure Technology Research and Management, India}
\address[4]{Department of Computer Engineering, Ain Shams University, Egypt}
\address[5]{Department of Computer Engineering, Technical University of Munich, Germany}

\begin{abstract}
We propose two distributed set-based observers using strip-based and set-propagation approaches for linear discrete-time dynamical systems with bounded modeling and measurement uncertainties. Both algorithms utilize a set-based diffusion step, which decreases the estimation errors and the size of estimated sets, and can be seen as a lightweight approach to achieve partial consensus between the distributed estimated sets. Every node shares its measurement with its neighbor in the measurement update step. In the diffusion step, the neighbors intersect their estimated sets using our novel lightweight zonotope intersection technique. A localization example demonstrates the applicability of our algorithms. 
%
%We propose two distributed set-based observers using strip-based and set-propagation approaches for linear discrete-time dynamical systems with bounded modeling and measurement uncertainties. Both algorithms utilize a set-based diffusion step, which is a new lightweight zonotope intersection technique. Our set-based diffusion step decreases the estimation errors and the size of estimated sets, and can be seen as a lightweight approach to achieve partial consensus between the distributed estimated sets. Every node shares its measurement with its neighbor in the measurement update step. The neighbors intersect their estimated sets constituting our proposed set-based diffusion step. A localization example demonstrates the applicability of our algorithms. 
%
\end{abstract}


\maketitle

\section{Introduction}

%State estimation across wireless networks is often a key technology for emergency rescue missions, homeland security \citep{conf:homeland}, habitat monitoring \citep{conf:habitat}, and home automation services \citep{conf:homeauto}. 

%State estimation algorithms are categorized according to the type of modeled uncertainties, which are either stochastic or set-based approaches. 
State estimation algorithms compute a single state, a probability distribution of the state, or a set of all possible states. In stochastic approaches, measurement and process noise are modeled by probability  distributions (e.g., Gaussian \citep{conf:kfintro}). On the other hand, set-based approaches assume noise to be unknown but bounded by known bounds. Safety-critical applications require guarantees on the state estimation during operation -- such guarantees can be provided by set-based approaches. Set-based approaches are traditionally used in fault detection by generating an adaptive threshold to check the consistency of the measurements with the estimated output set \citep{conf:set_fault1,combastel2015merging,conf:set_fault3,conf:actuator_fault,conf:fdi}. According to the terminology in \citep{conf:althoffJagat}, there are three types of set-based observers: interval-based observers, set-propagation observers, and strip-based obser-vers. We focus on the following literature survey on set-propagation and strip-based observers as they are the paper's main focus.

%Interval-based and strip-based observers have been introduced separately, and 

%They are very useful for fault detection \citep{}

%While some observers neglect uncertainties (e.g., Luenberger observers \citep{conf:luenberger}), other approaches use stochastic to deal with uncertainties. Since observers considering stochastic processes are well-known (see e.g., \citep{conf:kfintro}), 
% Thus, many studies targeted the set-based state estimation, which becomes quite active in the past decades. 

%

%The most popular one is the statistical modeling of uncertainties. Kalman filter \citep{conf:kfintro} belongs to this category. Related work on this category puts statistical distribution assumption about the uncertainties. Assuming Gaussian distribution is the most widespread assumption on this category.  

%They make use of observer gain in an explicit way in order to take the effect of measurements in the state estimation operation . 

%Related work in \citep{conf:interval2asym} is constructed for time-invariant exponentially stable linear systems with additive disturbances and ensures that the error dynamics satisfy some conditions
%Authors in \citep{conf:interval4} used differential inequalities that allow deriving bounds for continuous differential systems at any time instant. 

\textbf{Set-propagation observers:} They are based on a Luenberger observer and, in general, obtain possible sets of states by combining the model and the measurements through an observer gain \citep{conf:comparison}. By merging optimal and robust observer gain designs, a zonotopic Kalman filter (ZKF) is proposed in \citep{conf:gainoptimalityKalmanZonotope} based on the Frobenius norm of a zonotope. The same author proposed a joint zonotopic and Gaussian Kalman filter (ZGKF) in \citep{conf:mergingKFZonotopic} along with robust fault detection in the presence of both bounded and Gaussian disturbances. This line of work has been extended to nonlinear systems using a zonotopic extended Kalman filter (ZEKF) in \citep{conf:zonotopicEKF}.

%The upper and lower bounds are obtained, for instance, from differential equations as in \citep{conf:interval1,conf:interval3}. Related work in \citep{conf:interval2asym} designs an exponentially stable interval observer for a two-dimensional time-invariant linear system. The aforementioned work is extended for arbitrary finite dimensions in \citep{conf:interval4}. The previous works on linear systems have been extended to nonlinear systems in \citep{conf:interval7,conf:interval8}. Another observer was proposed based on Muller’s theorem for nonlinear uncertain systems in \citep{conf:interval6bioreactor}. Also, the authors in \citep{conf:h_inf} introduces $H_\infty$ design into interval estimation. Interval observers for uncertain biological systems are proposed in \citep{conf:interval5bio}.  Furthermore, an interval observer is designed in the form of partial differential equations for non-homogeneous heat equations under distributed measurements \citep{conf:heat}.

%, or there may be 

% encloses the real state if the unknown initial state of the real system can be bounded. They generally provide specific information about the system trajectory at any instant given the bounds on the initial conditions, i.e., noise
%Interval observers are useful for dealing with uncertainty in fault detection applications. 

%We can more specifically classify the related work to a strip-based based approaches and interval observer-based approaches. 

%They also over-approximate the estimated states considering bounded uncertainties. 
\textbf{Strip-based observers:} Unlike set-propagation observ-ers, which are based on observer gain derivation, strip-based observers generally intersect the set of states consistent with the model and the set consistent with the measurements to obtain the corrected state set \citep{conf:althoffJagat}. %In other words, they compute the set of admissible values for the state at each time instant.
One early example of strip-based observers is a recursive algorithm bounding the state by ellipsoids \citep{conf:1968}. Another example based on normalized least-mean-squares is presented in \citep{conf:setmem4LMS}. A strip-based state estimation algorithm based on  %the difference of convex functions programming  
DC programming is proposed in \citep{conf:setmem1dcprogramming}. Authors in \citep{conf:set_mem_discriptor} consider linear time-varying descriptor systems for strip-based estimation. Strip-baseds observers for nonlinear models are investigated in \citep{conf:setmem5,conf:setmem6,conf:setmem7,conf:setmem8,conf:set-mem-event}. They are also used in applications such as underwater robotics \citep{conf:setmem2water}, a leader following consensus problem in networked multi-agent systems \citep{conf:setmem3eventleader}, and localization \citep{conf:setloc}. Authors in \citep{conf:disevent} consider a class of discrete time-varying systems with an event-based communication mechanism over sensor networks. An interconnected multi-rate system is considered in \citep{conf:intermulti}. A strip-based filtering subject to replay attacks and quantization effects is considered in \citep{conf:Franklin_ReplyAttackDistributed} and \citep{conf:Franklin_quantization}, respectively. Also, a distributed strip-based estimation and formation control algorithm for a fleet of vehicles is proposed in \citep{conf:vehicleformations}, where the set-based estimation enclosure is guaranteed in spite of the lack of knowledge of the control signal applied by the rest of the vehicles.

%Strip-based with affine-projection is considered in \citep{conf:affineset}. Finally, nonlinear kernel adaptive filtering is proposed in \citep{conf:nonlinearset}.

%Different set representations of domains are used with strip-based based observers. For instance, ellipsoids are widely
%The symmetric properties of zonotopes help to reduce the computational load of using them iteratively. 

Different set representations have been used in set-based estimation, e.g., ellipsoids \citep{conf:ellipsoide,conf:set-ellipsoida,conf:dis_ellip_multirate},
orthotopes, and polytopes \citep{conf:orthotope,conf:polytope}. Zonotopes \citep{conf:zonotope_rep} are a special class of polytopes for which one can efficiently compute linear maps and Minkowski sums -- both are important operations for set-based observers. A strip-based observer based on zonotopes is introduced in \citep{conf:zono1_Combastel}. Another strip-based approach for disc-rete-time piecewise affine systems using zonotopes is studied in \citep{conf:zono2piecewise}. Yet another work considers discrete-time descriptor systems using zonotopes \citep{conf:zono_18pages}. Set-based estimation of uncertain discrete-time systems using zonotopes is also proposed in \citep{conf:zono4}. In our previous work \citep{conf:cdcdsse}, we considered secure state estimation with the aid of a diffusion step from the stochastic domain. In this work, we adapt the diffusion step to sets in order to guarantee state inclusion and safety. The measurement update step is common with our previous work \citep{conf:cdcdsse}, which is originally from \citep{conf:stripzono}.

%The zonotopic observer, in combination with Kalman filtering, is addressed in [6], [7]. 

%\textbf{Contributions:}
%In average consensus methods, the network must perform repeated averaging steps before arriving at a consensus. When these methods are used for distributed estimation, the nodes update their local estimates (using, for instance, a Kalman filter [6] or a leastsquares update [13]) and then they run a consensus step with several iterations to fuse the estimates. After the data have been fused, a new measurement is taken, the estimates are updated again, and so on. Thus, in consensus estimation algorithms there are two timescales: one for collecting the measurements, and one for running the averaging consensus.

%We consider the problem of distributed set-based estimation, where a set of nodes is required to collectively estimate the state set of a linear dynamic system from their measurements. 

% we 

\textbf{Contributions:} %Unlike centralized set-based estimation algorithms, which have potentially critical failure point at the central node, 
We propose two distributed set-based estimators, where a set of nodes is required to collectively estimate the set of possible states of a linear dynamical system in a distributed fashion. One main problem in distributed set-based estimation is the misalignment between the estimated sets by the distributed nodes. This problem is usually solved by consensus methods \citep{conf:dis-consRR}. However, traditional consensus methods require the sensor network to perform several iterations before arriving at a consensus, which causes a high overhead in set-based estimation. 
%In traditional distributed set-based estimation, every node in a sensor-network receives the estimates based on its measurements only; then, the node intersects its set with the estimated sets of its neighbors  \citep{conf:dis-interconnected,conf:dis-kalmaninspired,conf:dis-partial}. However, we propose to let every node shares its measurements with its neighbor for faster convergence. 
Unlike prior work and our previous work \citep{conf:cdcdsse}, we supplement our newly proposed estimators with a new set-based diffusion step which consists of a new lightweight zonotope intersection technique. We use the term \textit{diffusion} since our intersection formula resembles the traditional diffusion step in the stochastic Kalman filter. More specifically, our proposed diffusion step is considered an adaptation of the diffusion approach in the stochastic Kalman filter to sets in order to guarantee safety. We show that our diffusion step decreases the estimated sets' volume and can be seen as a lightweight approach for achieving partial consensus due to the agreement on the intersected sets. Furthermore, we provide closed forms for our parameter-finding optimization problems to realize faster execution tim-es. All used data and code to recreate our findings are publicly available\footnotemark.
%One only obtains a partial consensus using our algorithms because every node has different neighbors with different measurements; thus, the resulting sets do not fully agree. The measurement update step is common with our previous work \citep{conf:cdcdsse} which is from \citep{conf:stripzono}. 
 % and we do not require running consensus iterations. 

\footnotetext{\href{https://github.com/aalanwar/Distributed-Set-Based-Observers-Using-Diffusion-Strategies}{https://github.com/aalanwar/Distributed-Set-Based-Observers-Using-Diffusion-Strategies}}

The rest of the paper is organized as follows: the problem statement and preliminaries are in Section~\ref{sec:sysmodel}. In Section~\ref{sec:setmem}, we present the distributed strip-based diffusion observer as our first algorithm. Our second solution is the distributed set-propagation diffusion observer, which is introduced in Section~\ref{sec:berger}. An analytical analogy to the diffusion Kalman filter is provided in Section~\ref{sec:analogy}. Both algorithms are evaluated in Section~\ref{sec:eval}. Finally, we conclude the paper in Section~\ref{sec:conc}.

%The actual state of the plant must be confined within the estimation and prediction sets at every sampling period:
\begin{figure*}
\begin{subfigure}{.5\textwidth}
  \centering
  \includegraphics[width=.85\linewidth]{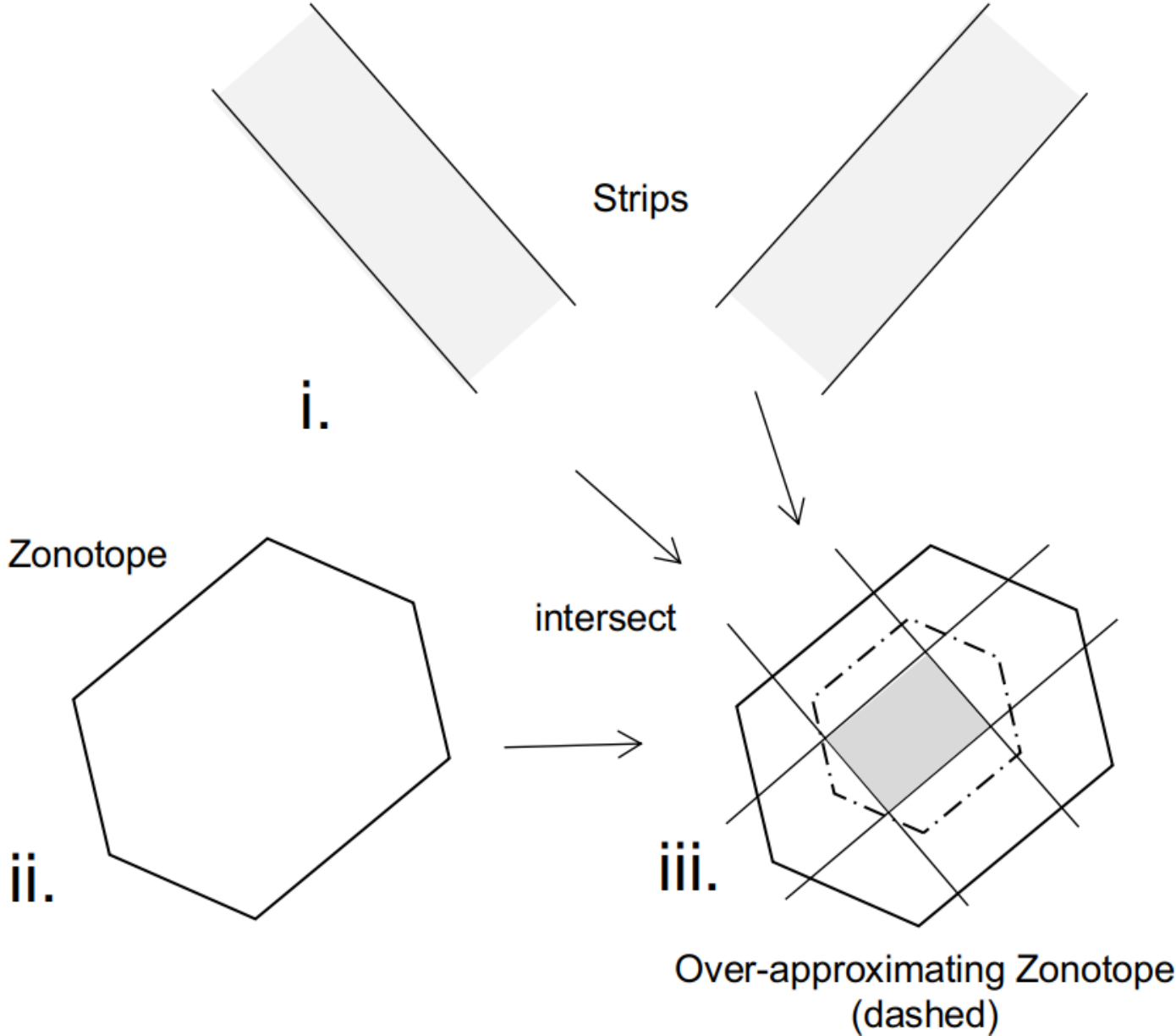}
  \caption{Measurement update step.}
  \label{fig:meas}
\end{subfigure}%
\begin{subfigure}{.5\textwidth}
  \centering
  \includegraphics[width=.6\linewidth]{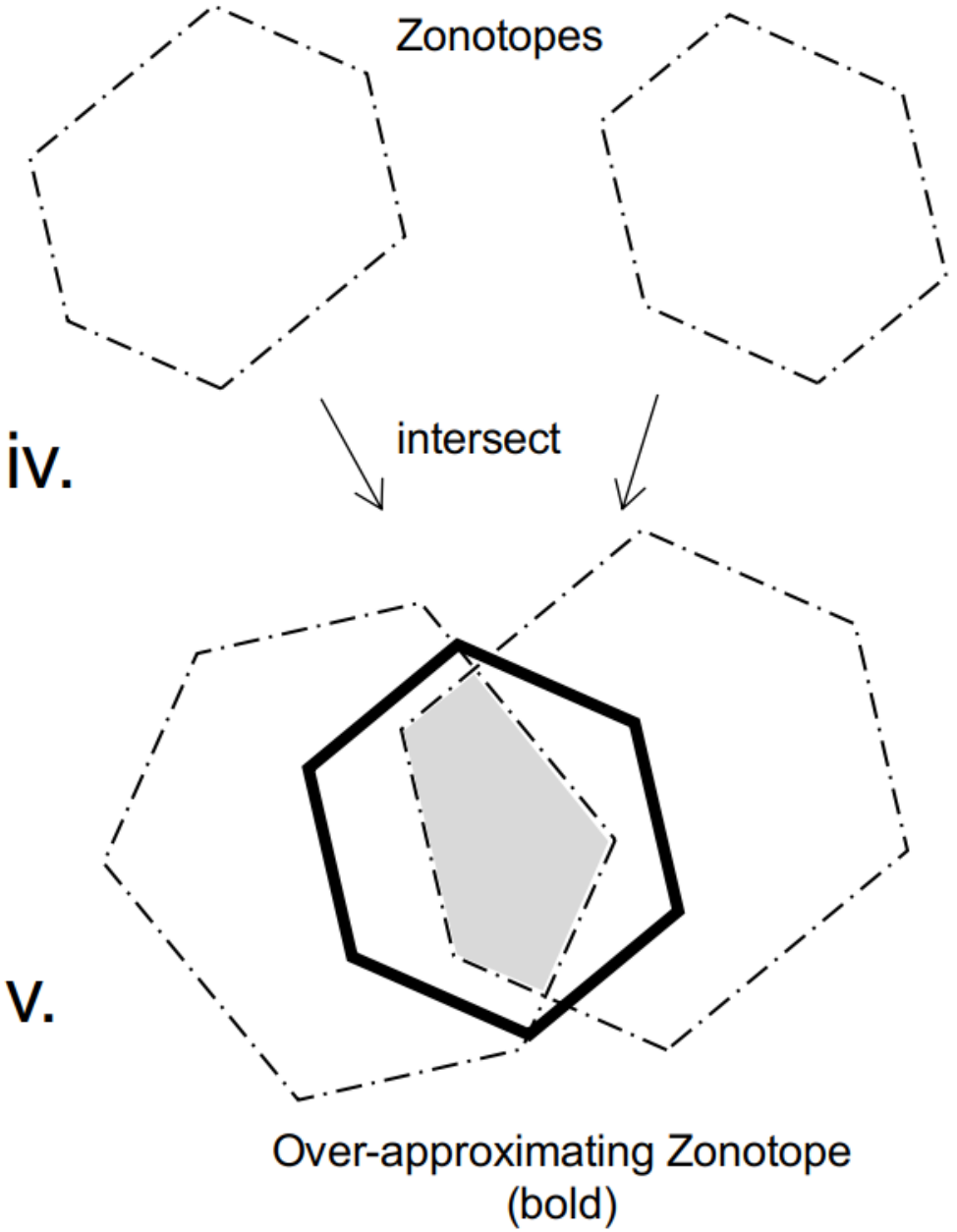}
  \caption{Diffusion update step. }
  \label{fig:diff}
\end{subfigure}
\caption{Strip-based approach in Algorithm~\ref{alg:1}, where (a) illustrates Lemma~\ref{th:stripzono} and (b) illustrates Theorem~\ref{th:diff}.} %Geometric shapes intersections of 
%\caption{Intersections for the set-membership approach in Algorithm~\ref{alg:1}. Figure~\ref{fig:meas} represents the measurement update step which consists of intersecting strips with zonotope. The resulting over-approximated zonotope (dashed) using Proposition~\ref{th:stripzono} is presented in Subfigure iii. Figure~\ref{fig:diff} illustrates the diffusion update step, where Subfigure v. shows the proposed over-approximated zonotope (bold) which is the intersection of two zonotopes using Theorem~\ref{th:diff}.}
%\caption{Geometric shapes intersections of set-membership approach in Algorithm~\ref{alg:1}. Figure~\ref{fig:meas} represents measurement update step where I. represents measurement state-sets (strips), II. shows the predicted state sets, and III. represents the over-approximated zonotope (dashed) of the intersection of the strips and zonotope. This has been formulated in Theorem~\ref{th:stripzono}. Figure~\ref{fig:diff} represents the diffusion update step where IV. represents the resultant zonotopes of the measurements update steps from two nodes, and V. represents the proposed over-approximation (bold) of the zonotopes intersection by diffusion. This has been formulated in Theorem~\ref{th:diff}.}
\label{fig:alg1}
%\vspace{2mm}
\end{figure*}

\section{Problem Statement and Preliminaries} \label{sec:sysmodel}
We start by stating some preliminaries before describing our problem statement.

%We start by stating some definitions around the proposed algorithms. Then we present our system model. 
%The distributed processing aims to get rid of the fusion center in the centralized version. It also aims to decrease the risk of a single point of failure and decrease the communication cost. 
%\citep{conf:kalman,conf:kfintro}

%\subsection{Preliminaries}

\begin{definition}
% \subset \R^n
\textbf{(Zonotope \citep{conf:zono1998})} A zonotope $\mathcal{Z}$ consists of a center $c \in \R^n$ and a generator matrix $G$ $\in$ $\R^{n \times e}$. We compose $G$ of $e$ generators $g^{(i)} \in \R^n$, $i=1,\dots ,e$, where $G=\begin{bmatrix}g^{(1)}, \dots, g^{(e)}\end{bmatrix}$. %\citep{conf:zonolec,conf:zonotope_rep}.
\begin{equation}
     \mathcal{Z} = \Big\{ c + \sum\limits_{i=1}^{e} \beta^{(i)}  g^{(i)} \Big| -1\leq \beta^{(i)} \leq 1 \Big\}.
    \label{equ:zonoDef}
\end{equation}
We use the shorthand $\mathcal{Z}= \zono{ c,G }$ for a zonotope. 
%\hfill $\Box$
\end{definition}

Given two zonotopes $\mathcal{Z}_1=\langle c_1,G_1 \rangle$ and $\mathcal{Z}_2=\langle c_2,G_2 \rangle$, the following operations can be computed exactly \citep{conf:zono1998}:
\begin{enumerate}
    \item Minkowski sum:
    \begin{equation}
     \mathcal{Z}_1 \oplus \mathcal{Z}_2 = \Big\langle c_1 + c_2, [G_1, G_2]\Big\rangle.
     \label{eq:minkowski}
     \end{equation}
    
    \item Linear map:
    \begin{equation}
     L \mathcal{Z}_1  = \Big\langle L c_1, L G_1\Big\rangle. 
     \label{eq:linmap}
     \end{equation}    
\end{enumerate}
 
% \textbf{Definition 2: (Strip)} A strip $S^{k,j}_{i}$is defined as

%\subsection{System Model} \label{sec:system}

 %We define the reachable set as the set of possible states $x_k$ which can be reached at each time step. In this work, reachable sets are represented by zonotopes due to their favorable computational complexity as discussed in \citep{conf:althoffthesis}. 
 Let $C \in \R^{n \times p}$, then $\|C\|_F=\sqrt{\text{tr}(C^TC )}$ is the Frobenius norm of $C$. The Frobenius norm of a vector $x \in \R^n$ equals the Euclidean norm of the vector defined as $\|x\|=\sqrt{x^Tx}$. The F-radius of the zonotope $\mathcal{Z}=\langle c,G \rangle$ is the Frobenius norm of the generator matrix. We denote the reduction operator by $\downarrow_q G$ of a generator matrix $G$. It basically reduces the number of generators of a zonotope to a fixed number $q$ %while preserving the inclusion property 
 so that the resulting zonotope is an over-approximation \citep{conf:reduceMethods}. The operator diag($x$) returns a diagonal matrix with $x$ on the diagonal. Finally, for a scalar $c$ and matrices $A$, $B$ and $C$, we provide the following trace properties \citep[p.11]{conf:cookbook}, where $\nabla_X f(X)$ is the derivative of $f(X)$ with respect to $X$:
 \begin{align}
    & \text{tr}(cA)=c\ \text{tr}(A), \label{eq:trscalMat} \\
    & \text{tr}(A+B)=\text{tr}(A)+\text{tr}(B), \label{eq:trMatpMat} \\
    &\nabla_X \text{tr}(A X B X^T C)= A^T C^T X B^T + C A X B, \label{eq:axbxc}\\
    & \nabla_X \text{tr}(B^T X^T C X B)= C^T X B B^T + C X B B^T. \label{eq:bxc}   
 \end{align}
 
We aim to estimate the set of possible states in a distributed fashion starting from the initial set $\mathcal{Z}_{0}$ by observing physical signals through sensory devices. Consider a set of $N$ nodes indexed by $i \in \{0, \ldots, N-1\}$ distributed geographically over some region. We denote the neighborhood of a given node $i$ by the set $\mathcal{N}_i$ containing $m_i$ nodes connected to node $i$, including the node itself. Every node is interested in estimating the set of possible states of the network state. The noise is assumed to be is unknown but bounded by a known bound and the initial set $\mathcal{Z}_{0}$ is known. We consider an observable discrete-time linear system model:
%\vspace{-5mm}
\begin{equation}
\begin{split}
 x_{k+1} &= F x_k + n_k,\\
y^{i}_k &= H^i x_k + v_k^i,
\end{split}
\label{eq:sysmodel}
\end{equation}
%\vspace{-4mm}
%The state of the network at step $i$ is $x_{i}$. 
where $x_k \in \mathbb{R}^{n}$ is the state at time step $k$, $y^i_k \in \mathbb{R}^{p}$ the measurement observed by node $i$ at time step $k$, $F \in \mathbb{R}^{n \times n}$ the state matrix, $H^i \in \mathbb{R}^{p \times n}$ the measurement matrix, $n_k$ the process noise, and $v_k^i$ the measurement noise. The process and measurement noises are assumed to be unknown but bounded by zonotopes: $n_k \in \mathcal{Z}_{Q,k}=\langle0,Q_k\rangle$ and $v_k^i \in \mathcal{Z}_{r,k}^i= \langle0,\text{diag}(r^i_k)\rangle$. If these zonotopes are not centered around zero, we perform a coordinate transformation. All vectors and matrices are real-valued and have proper dimensions.

%Let $\hat{x}^i_{i|j}$ denotes the estimate of $x_k$ at time step $i$ of node $i$ given observations up to and including time step $j$.

%\input{Sections/3-KF.tex}
%%------------------- manual fix for equation numbers
%%\setcounter{equation}{16}
%%----------------------------------------------------
%\stackrel{\eqref{equ:bj}}{=}

\section{Distributed Strip-based Diffusion Observer} \label{sec:setmem}

%strip-based geometric shapes intersection inspires the proposed algorithm. 
As mentioned in the introduction, we focus on two types of set-based observers: strip-based observers and set-propagation observers. We propose two algorithms extending the related work of both observers and adding the set-based diffusion step to both observers. Our contribution to strip-based observers is presented first. We denote the state estimated at node $i$ of the strip-based approach by $\hat{x}_{s,k}^i$ for time step $k$. The set of possible states in strip-based approaches are generally obtained from predicted, measurement, and corrected state sets, which are defined as follows:
 
%\begin{definition} \textbf{(Predicted State Set)} Given system \eqref{eq:sysmodel} with initial set $\mathcal{Z}_{s,0}= \langle c_{s,0},G_{s,0} \rangle$, the reachable set of states $\hat{\mathcal{Z}}_{s,k}^i$ of node $i$ is defined as the set of all possible solutions of $\hat{x}_{s,k}^i \in \hat{\mathcal{Z}}_{s,k}^i$ given set of all possible solutions of $\hat{x}_{s,k-1}^i \in \hat{\mathcal{Z}}_{s,k-1}^i$ and $\mathcal{Z}_{Q,k}$ which is the zonotope which bounds modeling noise \cite[p.4]{conf:zono_18pages}:
\begin{definition} \label{def:predset}
\textbf{(Predicted State Set)} Given the system in \eqref{eq:sysmodel} with initial zonotope $\hat{\mathcal{Z}}^i_{s,0}=\mathcal{Z}_{0}$, then the predicted reachable set $\hat{\mathcal{Z}}_{s,k}^i$ considering the zonotope  $\mathcal{Z}_{Q,k}$ which bounds the modeling noise is defined as:
\begin{equation}
   % Z_k = \{ x_k | x_k \in F Z_{i-1} + \langle 0,Q,i\rangle  \}.
\hat{\mathcal{Z}}_{s,k}^i= F \hat{\mathcal{Z}}_{s,k-1}^i \oplus \mathcal{Z}_{Q,k}. %\langle 0,Q,i\rangle
\end{equation}
\end{definition}
%\hfill $\Box$

\begin{definition} \label{def:measset} \textbf{(Measurement State Set)} Given the system in \eqref{eq:sysmodel}, then the measurement state set $\mathcal{S}^i_{k}$ of node $i$ is defined as the set of all possible solutions $x_{k}$ given $y_k^i$ and $v_k^i \in \mathcal{Z}_{r,k}^i= \langle0,\text{diag}(r^i_k)\rangle$. If the dimension of $y_k^i \in \mathbb{R}^{p}$ equals one, i.e., $p=1$, this measurement set is a strip:
\begin{equation}
    \mathcal{S}^i_k = \Big\{ x_k \Big| | H^i x_k - y^i_k| \leq r^i_k \Big\}, \label{eq:strip}
\end{equation}
and $\mathcal{S}^i_k$ is the intersection of multiple strips for $p>1$. 
\end{definition}
%\hfill $\Box$
%We represent this set using zonotope representation. 
\begin{definition} \label{def:corrset} \textbf{(Corrected State Set)} Given the system in \eqref{eq:sysmodel} with initial set $\mathcal{Z}_{0}$, then the reachable corrected state set $\bar{\mathcal{Z}}_{s,k}^i$ of node $i$ is defined as the over-approximation of the intersection between $\hat{\mathcal{Z}}_{s,k}^i$ and $\mathcal{S}^i_k$:
\begin{equation}
     \big( \hat{\mathcal{Z}}_{s,k}^i \cap \mathcal{S}^i_k \big) \subseteq \bar{\mathcal{Z}}_{s,k}^i. 
\end{equation}
\end{definition}
Our proposed strip-based approach consists of three steps: measurement update, diffusion update, and time update. The time update step results in computing the predicted set (Definition~\ref{def:predset}). The measurement update step utilizes the previously predicted set along with the measurement set (Definition~\ref{def:measset}) to compute the corrected set (Definition~\ref{def:corrset}). Every node in a distributed setting has access to some, not all, measurements. %Thus, every node would have its point of view of the estimated set, and this results in disagreement in many applications such as fault detection. In other words, it is essential to consense on the estimated set between the distributed nodes. 
Therefore, we propose sharing measurements and estimated sets in the measurement and diffusion update steps, respectively, to obtain a lightweight consensus between the distributed nodes. We first give a high-level description of the proposed algorithm in Algorithm~\ref{alg:1}, then we derive the required theory. Our approach corrects the reachable set of each node by determining the set of consistent states with the model and measurements received from all neighbors. More specifically, during the measurement update, every node collects measurements from neighbors, as shown in step i in Figure~\ref{fig:meas}, i.e., each node obtains a family of strips (measurements) to be intersected with the predicted zonotopic set (step ii in Figure~\ref{fig:meas}) of each node to obtain the estimated zonotope $\bar{\mathcal{Z}}^i_{s,k}$, dashed in  step iii in Figure~\ref{fig:meas}. Every node collects the shared sets from its neighbors in step iv in Figure~\ref{fig:diff}. Next, each node intersects its reachable set with shared sets of the neighbors in the set-based diffusion step in step v in Figure~\ref{fig:diff}. Finally, the estimated sets evolve according to the time update model. 

Let $\Gamma^i = \begin{bmatrix} H^{1_{}^T} ,\ldots, H^{m_{i}^T} \end{bmatrix}^T$, $\bar{y}^i_k=\begin{bmatrix} y^{{1}^T}_k ,\ldots, y^{m_{i}^T}_k \end{bmatrix}^T$, and $\bar{v}^i_k=\begin{bmatrix} v^{{1}^T}_k ,\ldots, v^{m_{i}^T}_k \end{bmatrix}^T \in \mathcal{Z}_{R,k}= \langle0,R^i_k\rangle$, where $R^i_k = \text{diag}\Big(\begin{bmatrix} r^{{1}^T}_k ,\ldots, r^{m_{i}^T}_k \end{bmatrix}^T\Big)$ with $m_i$ equals the number of available measurements from the node's neighbors. We propose to perform the measurement update step according to the following lemma \citep{conf:stripzono}, which is represented graphically in Figure~\ref{fig:meas}:
%We borrow the diffusion concept from distributed diffusion Kalman filter \citep{conf:diffusion} to enhance the estimate of each node. 
%Set membership can used to detect attacks \citep{conf:amrcdc2019}.

%\subsection{Algorithm}

 %Given system \eqref{eq:sysmodel} with the current estimated zonotope $Z^i_{i-1}=< C^i_{i-1},G^i_{i-1} >$, the family of $m$ measurements from neighbours restricts the state in strips \eqref{eq:strip}. The vectors $\lambda^{i,j}_{s,k} \in \R^n$. The corrected zonotope $Z_{\psi_k}=<\hat{C}_{\psi_k}, \hat{G}_{\psi_k}>$ based on the measurements y's is approximated by

 %$\hat{\mathcal{Z}}^j_{\psi_k}=<\hat{c}^j_{\psi_k}, \hat{G}^j_{\psi_k}>$
%$\Lambda^i_{s,k} =\Big[ \lambda_{s,k}^{i,1} \,\, \lambda_{s,k}^{i,2} \ldots  \lambda_{s,k}^{i,m_i}\Big]$
 \begin{lemma}[\citep{conf:stripzono}]
 \label{th:stripzono}
 Given are the zonotope $\hat{\mathcal{Z}}_{s,k-1}^i= \big\langle \hat{c}^i_{s,k-1},$ \\ $\hat{G}^i_{s,k-1} \big\rangle$, the family of $m_i$ measurement sets $\mathcal{S}^i_k$ in \eqref{eq:strip}, and the design parameters $\Lambda^i_{s,k}$. The intersection between the zonotope and measurement sets can be over-approximated by a zonotope $\bar{\mathcal{Z}}^i_{s,k}=\big\langle  \bar{c}^i_{s,k},\bar{G}^i_{s,k}\big\rangle $, where
  \begin{align}
  \bar{c}^i_{s,k} &=  \hat{c}^i_{s,k-1} +  \Lambda^i_{s,k}(\bar{y}^i_k - \Gamma^i \hat{c}^i_{s,k-1}), \label{eq:C_lambda}\\
  \bar{G}^i_{s,k} &= \Big[ (I - \Lambda_{s,k}^{i} \Gamma^i ) \hat{G}^i_{s,k-1}, \Lambda_{s,k}^i R^i_k \Big]. 
  \label{eq:G_lambda}
 \end{align}
%  \begin{align}
%   \bar{c}^i_{s,k} &=  \hat{c}^i_{s,k-1} + \sum\limits_{j\in\mathcal{N}_i} \lambda_{s,k}^{i,j}(y^{j}_k - H^j \hat{c}^i_{s,k-1}), \label{eq:C_lambda}\\
%   \bar{G}^i_{s,k} &= \Big[ (I - \sum\limits_{j\in\mathcal{N}_i} \lambda_{s,k}^{i,j} H^{j} ) \hat{G}^i_{s,k-1}, \lambda_{s,k}^{i,1} r^1_k,\dots,\lambda_{s,k}^{i,m_i} r^{m_i}_k \Big]. 
%   \label{eq:G_lambda}
%  \end{align}
 %\textbf{Proof:}
  \end{lemma}

The matrix of design parameters $\Lambda^i_{s,k}$ minimizing the Frobenius norm of the zonotope is obtained as in \citep{conf:stripwithzono}.

\begin{algorithm}[t]
\caption{Distributed Strip-Based Diffusion Observer}
\label{alg:1}
\begin{algorithmic}
\State Start with initial zonotope $\bar{\mathcal{Z}}_{s,k}^i = \mathcal{Z}_{0}$ for all nodes, and compute at every node $i$ and time instant $k$ the following: 
\State \textbf{Step 1}: Measurement update:
\begin{align*}
\Lambda^i_{s,k} &= \argminB_{\Lambda^i_{s,k}} \| \bar{G}^i_{s,k}   \|_F \nonumber\\
 \bar{c}^i_{s,k} &=  \hat{c}^i_{s,k-1} +  \Lambda^i_{s,k}(\bar{y}^i_k - \Gamma^i \hat{c}^i_{s,k-1}) \nonumber\\
  \bar{G}^i_{s,k} &= \Big[ (I - \Lambda_{s,k}^{i} \Gamma^i ) \hat{G}^i_{s,k-1}, \Lambda_{s,k}^i R^i_k \Big] \nonumber
\end{align*}% \nonumber
%\State END\nonumber
%\State \ \ \ \ Reduce $\hat{\mathcal{Z}}^i_{s_{\psi_k}}$ order by Girard method \citep{conf:girard}. %\citep{conf:girard}\nonumber.
\State \textbf{Step 2}: Diffusion update:
\begin{align*}
\mathfrak{w}^i_k =&  \argminB_{\mathfrak{w}^i_k} \| \grave{G}_{s,k}^i \|_F.  \nonumber\\
\grave{c}^i_{s,k}=&\frac{1}{\sum\limits_{j\in\mathcal{N}_i}w^{i,j}_k}\sum\limits_{j\in\mathcal{N}_i}w^{i,j}_k\bar{c}_{s,k}^j\nonumber\\
\grave{G}_{s,k}^i =& \frac{1}{\sum\limits_{j\in\mathcal{N}_i}w^{i,j}_k}\big[w^{i,1}_k\bar{G}^1_{s,k} ,...,w^{i,m_i}_k \bar{G}^{m_i}_{s,k} \big]\\
\tilde{G}_{s,k}^i  =&   \downarrow_q {\grave{G}_{s,k}^i}
%\grave{G}_{s,k}^i &=& \hat{G}^i_{\psi_k}\nonumber
\end{align*}
\State \textbf{Step 3}: Time update:
\begin{align*}
\hat{c}_{s,k}^i =& F \grave{c}^i_{s,k}\nonumber\\
\hat{G}_{s,k}^i =& [F \tilde{G}_{s,k}^i,  Q _k] \nonumber
\end{align*}
\end{algorithmic}
\end{algorithm}

%\vspace{-5mm}
 
As previously mentioned, every node shares its corrected zonotope $\bar{\mathcal{Z}}^i_{s,k}=\langle \bar{c}^i_{s,k} , \bar{G}^i_{s,k}  \rangle$ with its neighbours during the set-based diffusion step. We find the intersection between the shared zonotopes using the following theorem: 

 \begin{theorem}
 \label{th:diff}
The intersection between $m_i$ zonotopes $\bar{\mathcal{Z}}^j_{s,k}=\zono{\bar{c}^j_{s,k},\bar{G}^j_{s,k}}$, $\forall j \in \mathcal{N}_i$, can be over-approximated by the zonotope $\grave{\mathcal{Z}}^i_{s,k}=\zono{\grave{c}^i_{s,k},\grave{G}_{s,k}^i}$ with
 \begin{eqnarray}
\grave{c}^i_{s,k}&=&\frac{1}{\sum\limits_{j\in\mathcal{N}_i}w^{i,j}_k}\sum\limits_{j\in\mathcal{N}_i}w^{i,j}_k\bar{c}_{s,k}^j,\label{eq:cdiff}\\
\grave{G}_{s,k}^i &=& \frac{1}{\sum\limits_{j\in\mathcal{N}_i}w^{i,j}_k}\big[w^{i,1}_k\bar{G}^1_{s,k} ,...,w^{i,m_i}_k\bar{G}^{m_i}_{s,k}\big], \label{eq:gdiff}
 \end{eqnarray}
 where $w^{i,j}_k$ is a weight such that $\sum\limits_{j\in\mathcal{N}_i}w^{i,j}_k \neq 0$.% and $w^{i,j}_k \neq 0, \forall j \in \mathcal{N}_i$.
 \end{theorem}

 \begin{proof}
%Given that the $m_i$ zonotopes are intersecting, 
We aim to find the zonotope which over-approximates the intersection. Let $\bar{x} \in (\bar{\mathcal{Z}}^1_{s,k} \cap \bar{\mathcal{Z}}^2_{s,k} \cap ... \cap \bar{\mathcal{Z}}^{m_i}_{s,k} )$ then $\bar{x}$ is within the zonotope defined in \eqref{equ:zonoDef}, i.e., we have $z^{j} \in [-1,1]$ for each zonotope $j$ such that%we have $z^1, ..., z^{m_i}$ such that
 \begin{equation}
 \bar{x} = \bar{c}^j_{s,k} + \bar{G}^j_{s,k}z^{j}.
  %x &=& \hat{c}^1_{\psi_k} + \hat{G}^1_{\psi_k} z^1 \nonumber\\
  %&&... \nonumber\\
  %x &=& \hat{c}^{m_i}_{\psi_k} + \hat{G}^{m_i}_{\psi_k}z^{m_i}  \nonumber
 \label{equ:zonox}
 \end{equation}
  By multiplying \eqref{equ:zonox} with $w^{i,j}_k$ and summing for all $m_i$ zonotopes, we obtain  %the previous equations by $w^{i,1},..,w^{i,{m_i}}$, respectively. Then, we obtain
 \begin{align}
 \bar{x} &= \underbrace{\frac{1}{\sum\limits_{j\in\mathcal{N}_i}w^{i,j}_k}\sum\limits_{j\in\mathcal{N}_i}w^{i,j}_k\bar{c}_{s,k}^j}_{\grave{c}^i_{s,k}} \nonumber\\
 & + \underbrace{\frac{1}{\sum\limits_{j\in\mathcal{N}_i}w^{i,j}_k}\big[w_k^{i,1}\bar{G}^1_{s,k} ,...,w_k^{i,{m_i}}\bar{G}^{m_i}_{s,k}\big]}_{\grave{G}^i_{s,k}}     \underbrace{\begin{bmatrix} z^1 \\ \vdots \\ z^{m_i} \end{bmatrix}}_\textbf{z}  \nonumber\\
 &= \grave{c}^i_{s,k}  + \grave{G}^i_{s,k}  \textbf{z}.%\text{\ \ \ \ \ \ \ \ \ \ \ \ \ \ \ \ \ \ \ \ \ \ \ }\Box \nonumber
\end{align}
Note that $\textbf{z} \in [-\textbf{1},\textbf{1}]$ as $z^{1},\dots,z^{m_i} \in [-1,1]$. Thus, the center and the generator of the over-approximating zonotope are $\grave{c}^i_{s,k}$ and $\grave{G}^i_{s,k}$, respectively.
\end{proof}
%It is worth mentioning that if $w^{i,j}_k$ are ones, we obtain the regular averaging and which is considered type of zonotope intersection over-approximation. 

%We find a suitable weights $w^{i,j}_k$, $\forall j \in \mathcal{N}_i$ which minimize the uncertainties by reducing the size of the intersection zonotope. As a measure of the zonotope size, we minimize the Frobenius norm of the generators $\argminB_{w} \| \grave{G}^i_{s,k}  \|_F$. 

%%%%%%%%%%%%%%%%%%%%%--------------------  w design -----------------------%%%%%%%%%%%%%%%%%%%%%%%%%%%%%%%%%
%%%%%%%%%%%%%%%%%%%%%--------------------  w design -----------------------%%%%%%%%%%%%%%%%%%%%%%%%%%%%%%%%%
%%%%%%%%%%%%%%%%%%%%%--------------------  w design -----------------------%%%%%%%%%%%%%%%%%%%%%%%%%%%%%%%%%
%%%%%%%%%%%%%%%%%%%%%--------------------  w design -----------------------%%%%%%%%%%%%%%%%%%%%%%%%%%%%%%%%%
%%%%%%%%%%%%%%%%%%%%%--------------------  w design -----------------------%%%%%%%%%%%%%%%%%%%%%%%%%%%%%%%%%
%%%%%%%%%%%%%%%%%%%%%--------------------  w design -----------------------%%%%%%%%%%%%%%%%%%%%%%%%%%%%%%%%%
%%%%%%%%%%%%%%%%%%%%%--------------------  w design -----------------------%%%%%%%%%%%%%%%%%%%%%%%%%%%%%%%%%

%Our zonotopes intersection takes linear time ($\mathcal{O}(n)$). 

The optimal design of the weight vector $\mathfrak{w}^i_k=\Big[w^{i,1}_k,\dots,$ $w^{i,m_i}_k\Big]$ can be chosen such that the size of the zonotope $\grave{\mathcal{Z}}^i_{s,k}=\langle \grave{c}^i_{s,k},\grave{G}_{s,k}^i \rangle$ is minimal. Using the Frobenius norm as an indicator of zonotopic size, we compute $\mathfrak{w}^i_k$ by solving 
\begin{equation}
    \mathfrak{w}^i_k =  \argminB_{\mathfrak{w}^i_k} \| \grave{G}_{s,k}^i \|_F. 
    \label{eq:argminw}
\end{equation}
Next, we add the constraint $\sum\limits_{j\in\mathcal{N}_i}w^{i,j}_k = 1$ in order to facilitate finding the optimal weights $w^{i,j}_k$ in the following proposition.
%The following proposition is proposed to compute the .  
\begin{proposition} For the over-approximated set $\grave{\mathcal{Z}}^i_{s,k}$ $= \langle \grave{c}^i_{s,k}$, $\grave{G}_{s,k}^i \rangle$ in Theorem~\ref{th:diff}, the optimal design parameters $w^{i,j}_k$ for \eqref{eq:argminw}, $\forall j \in \mathcal{N}_i$  where $\sum\limits_{j\in\mathcal{N}_i}w^{i,j}_k = 1$, can be obtained as: 
\begin{equation}
w^{i,j}_k = \frac{1}{\text{tr}\Big(\bar{G}^j_{s,k} \bar{G}_{s,k}^{j^T}\Big) \sum\limits_{r\in\mathcal{N}_i} \frac{1}{\text{tr}\Big(\bar{G}^r_{s,k}\bar{G}_{s,k}^{r^T}\Big)}}.
%\frac{\mathlarger{\prod}\limits_{r=1}^{m_i,l \neq j} trace\Big(\bar{G}^r_{s,k}\bar{G}^{r^T}_{s,k}\Big)}{\mathlarger{\sum}\limits_{l=1}^{m_i} \mathlarger{\prod}\limits_{r=1}^{m_i,r \neq l} trace\Big(\bar{G}^r_{s,k}\bar{G}_{s,k}^{r^T}\Big)}.
\label{eq:optimalw}
\end{equation}
\end{proposition}

\begin{proof}
The Frobenius norm of the generator matrix can be computed as follows:
\begin{eqnarray}
    \| \grave{G}_{s,k}^i \|^2_F &=& \text{tr}\Big(\grave{G}_{s,k}^i \grave{G}_{s,k}^{i^T}\Big) \nonumber\\
    & \stackrel{\eqref{eq:gdiff}}{=}& \text{tr}\Big( \sum\limits_{r\in\mathcal{N}_i} (w^{i,r}_k)^2 \bar{G}^r_{s,k} \bar{G}^{r^T}_{s,k} \Big) \nonumber\\
    &\stackrel{\eqref{eq:trMatpMat}}{=}& \sum\limits_{r\in\mathcal{N}_i}  \text{tr}\Big( (w^{i,r}_k)^2 \bar{G}^r_{s,k}  \bar{G}^{r^T}_{s,k}\Big) \nonumber\\
    & \stackrel{\eqref{eq:trscalMat}}{=}& \sum\limits_{r\in\mathcal{N}_i} (w^{i,r}_k)^2 \text{tr}\Big( \bar{G}^r_{s,k}  \bar{G}^{r^T}_{s,k}\Big). 
\end{eqnarray}

Let $\alpha_r = \text{tr}\Big( \bar{G}^r_{s,k}  \bar{G}^{r^T}_{s,k}\Big)$, therefore we obtain the following constrained optimization problem:
\begin{align}
&w^{i,j}_k = \argminB_{w^{i,j}_k} \sum\limits_{r\in\mathcal{N}_i} \alpha_r (w^{i,r}_k)^2,\nonumber\\
& \text{subject to}   {\ }  f(\mathfrak{w}^i_k)=\sum\limits_{r\in\mathcal{N}_i} w^{i,r}_k -1 = 0. \label{eq:constopt}
\end{align}
This can be solved by introducing the Lagrange multiplier $s$~\citep{conf:lagrange}. The Lagrangian function for \eqref{eq:constopt} is
\begin{eqnarray}
\mathcal{L}=  \sum\limits_{r\in\mathcal{N}_i} \alpha_r (w^{i,r}_k)^2- s \Bigg(\sum\limits_{r\in\mathcal{N}_i} w^{i,r}_k - 1  \Bigg).
%\mathcal{L}= \argminB_{w^{i,j}_k} \sum\limits_{r\in\mathcal{N}_i} \alpha_r (w^{i,r}_k)^2- s \Bigg(\sum\limits_{r\in\mathcal{N}_i} w^{i,r}_k - 1  \Bigg).
\end{eqnarray}
The necessary condition $\forall j \in \mathcal{N}_i$ for an extremum point is 
\begin{equation}
%\begin{split}
\nabla_{w^{i,j}_k} \mathcal{L} = 2 w^{i,j}_k \alpha_j - s =0.  %\\
%\nabla_{w^{2}} \mathcal{L} &=& 2 w^{2} \alpha_2 - s =0 \\
%\vdots&  \\
%\nabla_{\sB{w^{i,m_i}}} \mathcal{L} =& 2 w^{i,m_i}_k \alpha_{m_i} - s =0, 
%\end{split}
\label{eq:derv_w}
\end{equation}
% The constraint provides the last condition:
% \begin{equation}
% \nabla_{s} \mathcal{L} = \sum\limits_{r\in\mathcal{N}_i} w^{i,r}_k - 1 =0. \label{eq:derv_s}
% \end{equation}
%Inserting \eqref{eq:derv_w} in \eqref{eq:derv_s} results in:
Inserting \eqref{eq:derv_w} in \eqref{eq:constopt} results in:
%\begin{equation}
%\sum\limits_{r\in\mathcal{N}_i} \frac{s}{2\alpha_r} = 1. %\label{eq:derv_s}
%\end{equation}
%Thus, the value of $s$ as follows:
\begin{equation}
s = \frac{2}{\sum\limits_{r\in\mathcal{N}_i} \frac{1}{\alpha_r}}. \label{eq:s}
\end{equation}
Inserting \eqref{eq:s} into \eqref{eq:derv_w} results in:
\begin{equation}
w^{i,j}_k \stackrel{\eqref{eq:derv_w}}{=} \frac{s}{2\alpha_j} \stackrel{\eqref{eq:s}}{=} \frac{1}{\alpha_j \sum\limits_{r\in\mathcal{N}_i} \frac{1}{\alpha_r}}. \label{eq:w}
\end{equation}
It remains to check if the extremum point is a minimum \citep{conf:lagrangeHes}. 
% First, we find the Jacobian of $f(\mathfrak{w}^i_k)$ with respect to $w^{i,j}_k$:
% \begin{equation}
% %\nabla_{\mathfrak{w}^i_k} f(\mathfrak{w}^i_k) = [ 1 1 \dot 1]. \label{eq:jac}
% \nabla_{w^{i,j}_k} f(\mathfrak{w}^i_k) = 1, \ \  \forall j \in \mathcal{N}_i. \label{eq:jac}
% \end{equation} Then
We compute the bordered Hessian matrix $H^b$, while suppressing the indices $i$ and $k$, and denoting $\nabla_{w^j} X(\mathfrak{w})$ by $X_{w^j}$ and $\nabla_{w^j w^{m}} X(\mathfrak{w})$ by $X_{w^{j,{m}}}$ for simplicity:
\begin{align}
H^b & =  \begin{bmatrix} 0 & -f_{w^1}& \dots & -f_{w^{m}} \\
-f_{w^1}& \mathcal{L}_{w^{1,1}} & \dots & \mathcal{L}_{w^{1,{m}}}  \\
\vdots & \vdots &\ddots & \vdots\\
-f_{w^{m}}& \mathcal{L}_{w^{{m},1}} & \dots & \mathcal{L}_{w^{{m},{m}}} 
\end{bmatrix} \nonumber\\
&=
\begin{bmatrix} 0 & -1 & -1 & \dots & -1 \\
-1 & 2\alpha_1 & 0& \dots & 0 \\ 
-1 & 0  & 2\alpha_2 & \dots& 0 \\ 
\vdots &\vdots  & &\ddots &\\
-1 & 0 & 0 & \dots & 2\alpha_{m}
\end{bmatrix}.     
\label{eq:hes}
\end{align}
The $m - 1$ largest  principal  minors  of \eqref{eq:hes} are negatives because $\alpha_r$ is positive $\forall \, r=1,\dots,m$. Thus, the extremum in \eqref{eq:w} is a minimum point, which concludes the proof. %$\Box$
\end{proof}
After presenting our distributed strip-based approach using the diffusion strategy, we present our set-propagation diffusion observer.

\section{Distributed Set-Propagation Diffusion Observer}\label{sec:berger}
Unlike the strip-based observer developed in the previous section, which was based on geometric intersection, we propose the following set-propagation observer based on the following structure while bounding the unknown noise by the corresponding bounding zonotope: 
% \begin{equation}
% \label{Eq:LO}
% x_{k}^i = F x_{k-1}^i + n_k + \sum\limits_{j\in\mathcal{N}_i}  \lambda^{i,j}_{v,k} (y^{j}_k - H^j  x_{k-1}^i - v_k^j ),
% \end{equation}
\begin{equation}
\label{Eq:LO}
x_{k}^i = F x_{k-1}^i + n_k + \Lambda^i_{v,k}(\bar{y}^i_k - \Gamma^i x^i_{k-1} - \bar{v}^i_k) 
\end{equation}
 where $\Lambda^i_{v,k}$ is a time-varying observer gain computed at each time step. The design of the observer makes use of the bounds of the noises. Let $x_{v,k}^i$ be the state estimated by the set-propagation observer. % and $\Lambda^i_{v,k} =\begin{bmatrix} \lambda_{v,k}^{i,1}  &\ldots& \lambda_{v,k}^{i,m_i}\end{bmatrix}$.  % which minimize the resultant zonotope size and estimation error.
%For bounding the uncertain states of the system \eqref{eq:sysmodel}, the objective is to design the gains $ \lambda^{\s k,j}_{v,k}$ in order to reduce the estimation error. 
For the distributed system in \eqref{eq:sysmodel}, the proposed design consists of two steps: Luenberger update and diffusion update. During the Luenberger update, every node shares its measurement with its neighbour, while in the diffusion step, every node shares the estimated information with its neighbours. We first discuss the Luenberger update step.%, by representing the Luenberger observer \eqref{Eq:LO} in form of zonotopes in the following theorem. 
\begin{theorem}
 \label{th:berg}
 Given are the system in \eqref{eq:sysmodel}, the measurements $y^i_k$, several zonotopes bounding $\hat{x}^i_{v,0} \in \mathcal{Z}_0 $, $n_k \in \mathcal{Z}_{Q,k} =  \langle 0,Q_k \rangle $, $\bar{v}_k^i\in  \mathcal{Z}^i_{R,k}=\langle 0,R_k^i \rangle $, and the state $\hat{x}^i_{v,k-1} \in \Big\langle\hat{c}^i_{v,k-1}$,  $\hat{G}^i_{v,k-1}\Big\rangle $. The zonotope bounding the uncertain states can be iteratively obtained as $\hat{x}^i_{v,k} \in  \Big\langle\bar{c}^i_{v,k} , \bar{G}^i_{v,k} \Big\rangle$, where
 \begin{align}
  \bar{c}^i_{v,k}  =&  F\hat{c}^i_{v,k-1}+  \Lambda_{v,k}^{i} \Big( \bar{y}^i_k -\Gamma^i\hat{c}^i_{v,k-1}  \Big) \label{eq:Cf_lambda2},\\
  \bar{G}^i_{v,k}  =&  \Big[\Big(F -  \Lambda_{v,k}^{i} \Gamma^i\Big)\hat{G}^{i}_{v,k-1}, -\Lambda_{v,k}^{i} R_k^i ,Q_k\Big].
  \label{eq:Gf_lambda2}
 \end{align}
\end{theorem}
\begin{proof} Given $\hat{x}_{v,k-1}^i \in \Big\langle\hat{c}_{v,k-1}^i, \hat{G}_{v,k-1}^i\Big\rangle$, $n_k \in \mathcal{Z}_{Q,k}$ $= \langle0$ $,Q_k\rangle$, and $\bar{v}_k^i\in \mathcal{Z}_{R,k}^i=\langle 0,R_k^i \rangle $, and by using \eqref{Eq:LO}, one obtains:
\begin{eqnarray}
&&\hat{x}_{v,k}^i \in   \Big\langle\bar{c}_{v,k}^i , \bar{G}_{v,k}^i \Big\rangle \nonumber\\
&& \stackrel{\eqref{Eq:LO}}{=}\Big(F - \Lambda_{v,k}^{i} \Gamma^i \Big) \hat{\mathcal{Z}}_{v,k-1}^i \oplus \mathcal{Z}_{Q,k}  %\nonumber\\
 \nonumber\\
&&\quad  \oplus  \Big\langle \Lambda_{v,k}^{i} \bar{y}^i_k  , 0 \Big\rangle \oplus -  \Lambda^i_{v,k} \mathcal{Z}_{R,k} %(-\lambda^{i,1}_{v,k}) \mathcal{Z}_{R,k}^1 \oplus \dots  \oplus (-\lambda^{i,m_i}_{v,k}) \mathcal{Z}_{R,k}^{m_i} 
\nonumber\\
&&=\Big(  F - \Lambda_{v,k}^{i} \Gamma^i \Big) \Big\langle \hat{c}^i_{v,k-1},\hat{G}^i_{v,k-1} \Big\rangle \oplus \langle0,Q_k\rangle \nonumber\\
&&\quad \oplus  \Big\langle  \Lambda_{v,k}^{i} \bar{y}^i_k , 0 \Big\rangle\! 
 \oplus\! \Big\langle 0,-\Lambda_{v,k}^{i} R_k^i \Big\rangle
 \nonumber\\
&&\stackrel{\eqref{eq:minkowski},\eqref{eq:linmap}}{=}  \Bigg\langle \underbrace{\Big[ F\hat{c}^i_{v,k-1}+  \Lambda_{v,k}^{i} \Big( \bar{y}^i_k -\Gamma^i\hat{c}^i_{v,k-1}  \Big) \Big]}_{\bar{c}_{v,k}^i} ,\nonumber\\
&&\quad \underbrace{\Big[\Big(F -  \Lambda_{v,k}^{i} \Gamma^i\Big)\hat{G}^{i}_{v,k-1}, -\Lambda_{v,k}^{i} R_k^i ,Q_k\Big]}_{\bar{G}_{v,k}^i} \Bigg\rangle \nonumber
\end{eqnarray}
\end{proof}

\begin{algorithm}[t]
\caption{Distributed Set-Propagation Diffusion Observer}
\label{alg:2}
\begin{algorithmic}
\State Start with initial zonotope $\bar{\mathcal{Z}}_{v,k}^i = \mathcal{Z}_{0}$ for all nodes, and at every time instant $k$, compute at every node $i$: 
\State \textbf{Step 1}: Luenberger update:
\begin{eqnarray*}
 \Lambda^i_{v,k} &=& \argminB_{\Lambda^i_{v,k}} \| \bar{G}^i_{v,k}  \|_F \nonumber\\
  \bar{c}^i_{v,k}  &=& F\hat{c}^i_{v,k-1}+  \Lambda_{v,k}^{i} \Big( \bar{y}^i_k -\Gamma^i\hat{c}^i_{v,k-1}  \Big)\nonumber\\ 
  \bar{G}^i_{v,k}  &=&  \Big[\Big(F -  \Lambda_{v,k}^{i} \Gamma^i\Big)\hat{G}^{i}_{v,k-1}, -\Lambda_{v,k}^{i} R_k^i ,Q_k\Big] 
\end{eqnarray*} %\nonumber
%\State END\nonumber
%\State \ \ \ \ Reduce $Z^i_{v,k}$ order by Girard method. \citep{conf:reduceMethods}.
%\citep{conf:girard}\nonumber.
\State \textbf{Step 2}: Diffusion update:
\begin{eqnarray*}
\mathfrak{w}_k^i &=&  \argminB_{\mathfrak{w}_k^i} \| \tilde{G}_{v,k}^i \|_F. \\
\hat{c}^i_{v,k}&=&\frac{1}{\sum\limits_{j\in\mathcal{N}_i}w^{i,j}_k}\sum\limits_{j\in\mathcal{N}_i}w^{i,j}_k\bar{c}_{v,k}^j\nonumber\\
\tilde{G}_{v,k}^i &=& \frac{1}{\sum\limits_{j\in\mathcal{N}_i}w^{i,j}_k}\big[w^{i,1}_k\bar{G}^1_{v,k},\dots,w^{i,m_i}_k\bar{G}^{m_i}_{v,k}\big]\nonumber\\
  \hat{G}_{v,k}^i  &=&   \downarrow_q {\tilde{G}_{v,k}^i} 
\end{eqnarray*}
%\State \textbf{Step 3}: Time update:
%\begin{eqnarray}
%\hat{G}_{i+1}^i &=& \bar{F}_k \hat{G}_{v_{i|i}}^i + Q_k \nonumber\\
%\hat{C}_{i+1}^i &=& \bar{F}_k \hat{C}^i_{v_{i|i}}\nonumber
%\end{eqnarray}
\end{algorithmic}
\end{algorithm}

We propose to compute the design vectors $\Lambda^i_{v,k}$ according to \citep{conf:gainoptimalityKalmanZonotope} such that:
\begin{equation}
 \Lambda^i_{v,k} = \argminB_{\Lambda^i_{v,k}} \| \bar{G}^i_{v,k} \|_F,
 \label{eq:argminG_v}
\end{equation}
which is a lightweight indication of the volume to decrease the computation cost and maintain a good performance.
The following proposition is provided to compute the optimal parameters $\Lambda^i_{v,k}$.

% \begin{proposition} For the zonotope $\hat{\mathcal{Z}}^i_{\psi_k}$, the optimal parameters set $\Lambda^{\star} =\begin{bmatrix}\lambda^{1,j}_k&\lambda^{2,j}_k&\ldots& \lambda^{\s k,j}_{v,k}\end{bmatrix}$, which ensures its Frobenius norm is minimized can be computed as  
% \begin{equation}
%     \Lambda^{\star} = -\begin{bmatrix}F\hat{G}^i_{v,k-1}\hat{G}^i_{v,k-1}^T \Gamma_1^T + Q_kQ_k^T \Gamma_{(m+1)}^T\end{bmatrix} \begin{bmatrix}\Gamma_1\hat{G}^i_{v,k-1} \Gamma_1^T,\Gamma_2R_1R_1^T\Gamma_2^T, \ldots,    \Gamma_m R_m R_m^T\Gamma_m^T\end{bmatrix}^{-1},  
%     \end{equation}
% where, $\Gamma_1 =\begin{bmatrix} H^{i,1}_k& H^{i,2}_k&\ldots& H^{i,m}_k\end{bmatrix}^T $, $\Gamma_p =\begin{bmatrix}0_{(p-1) \times 1} & R_p & 0_{(m-1) \times 1}   \end{bmatrix}^T$, for $1\leq p \leq m$.
% \end{proposition}
% \begin{proof}
% Can be easily established similar to Proposition 1.
% \end{proof}
% Following the Luenberger update step, in the diffusion step each node shares the information of the estimated zonotope $\langle \hat{c}^i_{\psi_k}, \hat{G}^i_{\psi_k} \rangle$ with its neighbours. The intersection between the shared zonotopes is then computed as discussed earlier in Theorem 2. The iterative design of the above two-step Luenberger observer is provided in Algorithm 2. 
% }

\begin{proposition} For the estimated zonotopic set $\hat{\mathcal{Z}}^i_{v,k} = \langle \hat{c}^i_{v,k-1}$, $\hat{G}^i_{v,k-1} \rangle$ corresponding to node $i$, the optimal design parameters $\Lambda^i_{v,k}$ for \eqref{eq:argminG_v} can be obtained as:
\begin{equation}
\Lambda^i_{v,k} = \frac{F \hat{G}^{i^{}}_{v,k-1} \hat{G}^{i^T}_{v,k-1} \Gamma^{i^T}}{ \Gamma^i \hat{G}^i_{v,k-1}\hat{G}^{i^T}_{v,k-1} \Gamma^{i^T} +   R_k^{i} R_k^{i^T}  }.
\label{eq:lamboptF}
\end{equation}
%where $\Gamma^i = \begin{bmatrix} H^{1_{}^T} &\ldots& H^{m_{i}^T} \end{bmatrix}^T$.
\label{pr:lamboptF}
\end{proposition}
%\vspace{-4mm}
%\begin{proof}
%Can be easily established similar to Proposition~\ref{pr:lambopt}.$\Box$
%\end{proof}

%\vspace{-4mm}
\begin{proof}  
 The proof is along the lines of \citep{conf:stripwithzono}. 
%We rewrite \eqref{eq:Gf_lambda2} as
% \begin{equation}\label{eq:compact_g}
%  \bar{G}^i_{v,k} = \Big[\Big(F-\Lambda^i_{v,k} \Gamma^i\Big) \hat{G}^i_{v,{k-1}},~\Lambda^i_{v,k} R_k^i  \Big].
% \end{equation}
The Frobenious norm of \eqref{eq:Gf_lambda2} can be computed as
\begin{align}
\| \bar{G}^i_{v,k} \|^2_F =& \text{tr}\ (\bar{G}^{i^{}}_{v,k} \bar{G}^{i^T}_{v,k})\nonumber \\ 
 \stackrel{\eqref{eq:Gf_lambda2}}{=}& \text{tr}\ \Bigg( \Big(F-\Lambda^i_{v,k} \Gamma^i \Big) \mathcal{G}\Big(F-\Lambda^i_{v,k} \Gamma^i  \Big)^T \nonumber\\ 
& +  \Lambda^i_{v,k}  \mathcal{R}^i\Lambda^{i^T}_{v,k} + Q_kQ_k^T \Bigg) \nonumber\\ 
=&\text{tr}\Big( F\mathcal{G}F^T - F\mathcal{G} \Gamma^{i^T} \Lambda^{i^T}_{v,k}  - \Lambda^i_{v,k} \Gamma^i  \mathcal{G}F^T   \nonumber\\ 
& +  \Lambda^i_{v,k} \Gamma^i \mathcal{G} \Gamma^{i^T}\Lambda^{i^T}_{v,k} +  \Lambda^i_{v,k}  \mathcal{R}^i\Lambda^{i^T}_{v,k} {+} Q_kQ_k^T \Big),  \label{eq:ggt}
\end{align}
where $\mathcal{R}^i= R_k^i R_k^{i^T}$ and $\mathcal{G}=\hat{G}^{i^{}}_{s,{k-1}} \hat{G}^{i^T}_{s,{k-1}}$. The optimal value of $\Lambda^i_{v,k}$ can be obtained by solving
\begin{equation}
\nabla_{\sm \Lambda^i_{v,k}} \| \bar{G}^i_{v,k} \|^2_F = 0,
\label{eq:jacggt}
\end{equation}  
where the Jacobian in \eqref{eq:jacggt} can be computed by applying matrix properties in \eqref{eq:axbxc} and \eqref{eq:bxc} to \eqref{eq:ggt}:
\begin{eqnarray}
\nabla_{\sm \Lambda^i_{v,k}} \| \bar{G}^i_{v,k} \|^2_F 
&=& -2 F\mathcal{G} \Gamma^{i^T}  + 2 \Lambda^i_{v,k} \Gamma^i \mathcal{G} \Gamma^{i^T} +  2 \Lambda^i_{v,k}  \mathcal{R}^i \nonumber\\
&=& 0
\label{eq:jacg}
\end{eqnarray}
By inserting the optimal $\Lambda^i_{v,k}$ from \eqref{eq:lamboptF} in \eqref{eq:jacg}, one can see that \eqref{eq:jacg} is fulfilled. We do not have to check the type of extremum since we have an unconstrained norm, which is always convex. %$\Box$
%The optimal $\Lambda^i_{v,k}$ can be computed as given in \eqref{eq:lambopt} and this concludes the proof.
\end{proof}

Following the Luenberger update step, in the diffusion step, each node shares the information of the estimated zonotope $\langle \bar{c}^i_{v,k}, \bar{G}^i_{v,k} \rangle$ with its neighbours. The intersection between the shared zonotopes is then computed as discussed earlier in Theorem~\ref{th:diff}. The iterative design of the above two-step Luenberger observer is provided in Algorithm~\ref{alg:2}. We note that one of the differences between the introduced set-propagation observer and strip-based observer is the position of the time update step. This appears by comparing Theorem~\ref{th:berg} with Theorem~\ref{th:stripzono}. Also, it can be easily shown that 
\begin{align}
    \Lambda^i_{v,k} = F \Lambda^i_{s,k}.
    \label{eq:lambCompare}
\end{align}

\section{Comparison to Diffusion Kalman Filter}\label{sec:analogy}

In this section, we build the analogy with the diffusion Kalman filter (DKF) \citep{conf:diffusion}. We start by defining the Gaussian random vectors $n_k \sim \mathcal{N}(0,\bar{Q}_k)$ and $v_k^i \sim \mathcal{N}(0,\bar{R}_k^i)$ where $\mathcal{N}(\mu,\Sigma)$ denotes a Gaussian distribution with mean $\mu$ and covariance $\Sigma$. We start by showing the analogy with the incremental update, then we show it for the diffusion update. 
%Then, we define the innovation (error) term for DKF as:
%\begin{equation}\label{Eq.DKF_innovation}
%e_{k-1} = y^{j}_{k-1} - H^i {x}_{k-1}^i 
%\end{equation}
%We have the incremental update and then the diffusion update.
\subsection{Incremental Update}

Corresponding to incremental update in DKF we have for every node $i$ with a neighbor node $j$ \citep{conf:diffusion}
\begin{align*}
&{S}^i_{k}  = \bar{R}^i_k + H^i {P}^i_{k-1}H^{i^T},\\
&\bar{x}^i_{k} = \hat{x}^i_{k-1} + {P}^i_{k-1} H^{i^T}{S}_{k}^{i^{-1}}( y^{j}_{k} - H^i \hat{x}_{k-1}^i ), \\
&{P}^i_{k} = {P}^i_{k-1} - {P}^i_{k-1}H^{i^T}{S}_{k-1}^{i^{-1}}H^{i}{P}^i_{k-1}.
\end{align*}
Let us introduce 
\begin{align}
L_{k} = {P}^i_{k-1} H^{i^T} {S}_{k}^{i^{-1}}. 
\label{eq:LDKF}
\end{align}
Then after the incremental update, the center estimate corresponding to each node can be expressed using the gain $L_{k}$ as
\begin{align}
\bar{x}^i_{k} &=  \hat{x}^i_{k-1}+ L_{k}( y^{j}_{k} - H^i \hat{x}_{k-1}^i ). \label{eq:diffx}
\end{align}
%
%Similarly, the covariance term after the incremental update can be written as
%\begin{eqnarray}\nonumber
%{P}^i_{k-1} &=&  (I - L_{k-1}H^i_{k-1}){P}^i_{k-1}
%\end{eqnarray}
%
We have from \eqref{eq:lamboptF} and \eqref{eq:lambCompare}
\begin{equation}
\Lambda^i_{s,k} = \frac{\hat{G}^{i^{}}_{s,k-1} \hat{G}^{i^T}_{s,k-1} \Gamma^{i^T}_k}{ \Gamma^i_k \hat{G}^i_{s,k-1}\hat{G}^{i^T}_{s,k-1} \Gamma^{i^T}_k +   R_k^{i} R_k^{i^T}  },
\label{eq:lamboptFstrip}
\end{equation}

Comparing \eqref{eq:LDKF} with \eqref{eq:lamboptFstrip} and \eqref{eq:diffx} with \eqref{eq:C_lambda}  for a node $i$ with a single neighbor $j$ results in the following analogy
\begin{align}
   P^i_{k} \equiv \bar{G}^i_{s,k} \bar{G}_{s,k}^{i^T}, \quad 
   \bar{R}_k \equiv R_k^i R_k^{i^T}, \quad \bar{x}^i_{k} \equiv \bar{c}^i_{s,k}.
\end{align}

\subsection{Diffusion Update}

The diffusion update in DKF can be expressed as
\begin{eqnarray*}
\hat{x}^i_{k}&=&\sum\limits_{j\in\mathcal{N}_i}w_k^{i,j}\bar{x}_{k}^j\nonumber,
%\hat{x}^i_{k}&=&F\hat{x}^i_{k-1}\nonumber\\
%{P}^i_{k} &=& F{P}^i_{k-1}F^T + Q_iQ_i^T
\end{eqnarray*}
which resembles our zonotopes intersection formula in \eqref{eq:cdiff} in case of $\sum\limits_{j\in\mathcal{N}_i} w_k^{i,j} =1$. Moreover, the optimal weight based on covariance intersection \citep[Section III]{conf:covintersection} is a function of $\frac{1}{tr(P^j_{k}) \sum\limits_{r\in\mathcal{N}_i} \frac{1}{tr(P^r_{k})}}$ which again resembles our optimal weights in \eqref{eq:optimalw}.

%\begin{align}
%    w^{i,j}_k = \frac{1}{tr(P^j_{k}) \sum\limits_{r=1}^{m_i} \frac{1}{tr(P^r_{k})}}
%\end{align}

%After the diffusion step is done, the estimated center can be re-written as
%\begin{eqnarray*}
%\hat{x}^i_{k} &=& F\hat{x}^i_{k-1}\\
%  &=& F [{x}^i_{k-1} + L^*_{k-1}*e_{k-1}] \\
%  &=& (F-M^*_{k-1}H^i_{k-1})\hat{x}^i_{k-1} + M^*_{k-1}y^{j}_{k-1}
%\end{eqnarray*}
%with $M^*_{k-1} = FL^*_{k-1}$. Similarly the covariance can be given as
%\begin{equation}\label{Eq.Cov_DKF}
%\hat{P}^i_{k} = F{P}^i_{k-1}F^T + Q_iQ_i^T
%\end{equation}
%----------------------
\subsection{Discussion}

Based on the above results, it can be said that if the initial center $c_0$ is chosen equally for the DKF and our strip-based observer along with the same initial covariance matrices $P_0  =  \hat{G}_{0}^i(\hat{G}_{0}^i)^T $ and same diffusion weights, then in the absence of any over-approximation error due to the zonotopic reduction operation, both observers would result in the same optimal gain  $L_k$ and return the same center estimate $\hat{x}^i_{k}$.  %Note that in the presented earlier design of the DKF for the diffusion step it was assumed $\sum\limits_{j\in\mathcal{N}_i}w^{i,j} = 1$. In contrast in our proposed design there is no such restriction on the design weights while ensuring that chosen weights are a convex combination. %Thus, for the case when  $\sum\limits_{j\in\mathcal{N}_i}w^{i,j} = 1$, the proposed design in our paper would result in same design gains as presented in the DKF works earlier. 

\begin{figure*}[h]
%\vspace{-0.05cm}
    \centering
    \begin{tabular}{ p{0.50\textwidth}  p{0.50\textwidth}}
        \resizebox{0.49\textwidth}{!}{
            \begin{subfigure}[h]{0.49\textwidth}
      \centering
        \includegraphics[scale=0.4]{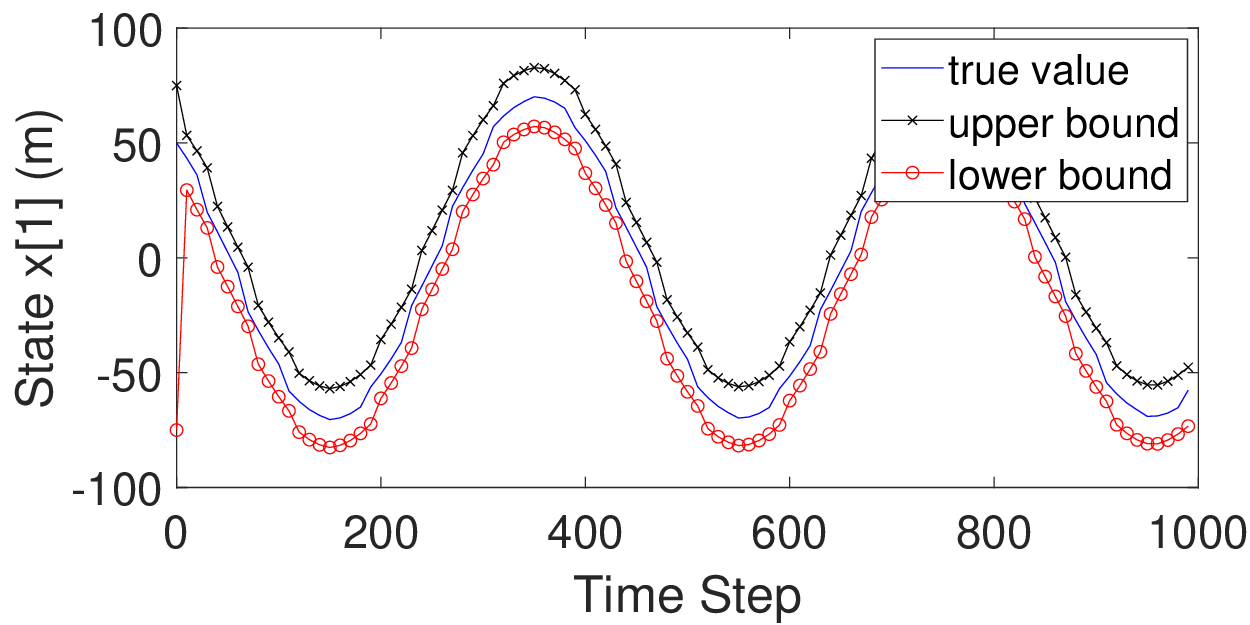}
        \caption{}
        \label{fig:set-x}
    \end{subfigure}
       } 
   &
   \resizebox{0.49\textwidth}{!}{
            \begin{subfigure}[h]{0.49\textwidth}
      \centering
        \includegraphics[scale=0.4]{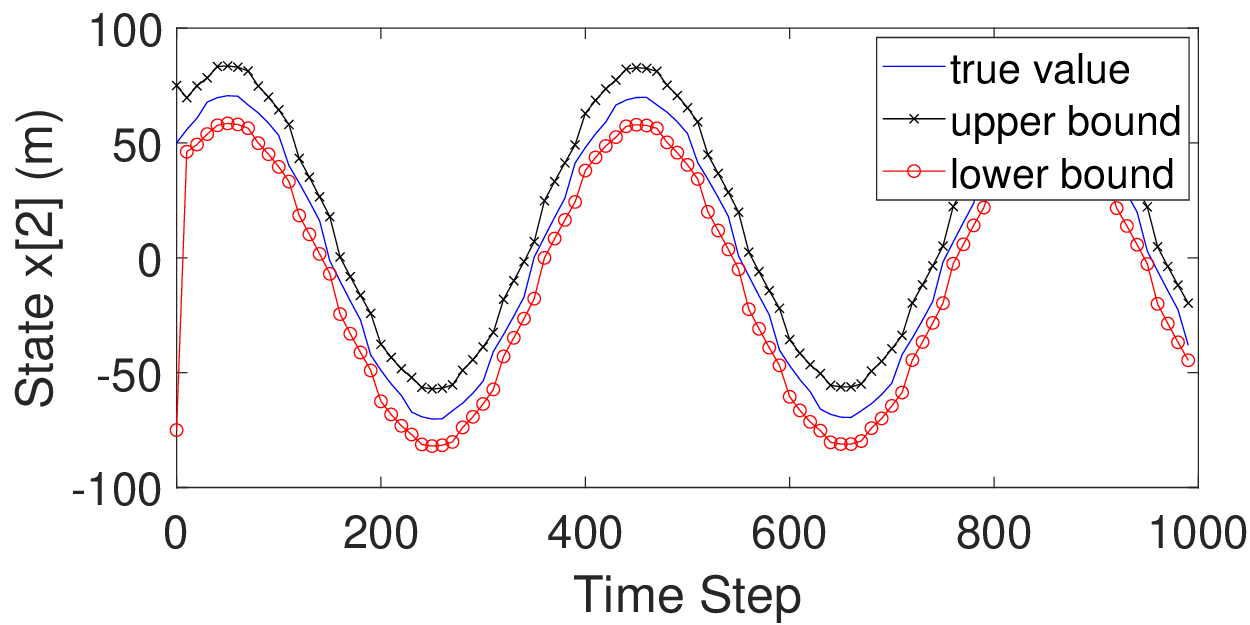}
        \caption{}
        \label{fig:set-y}
    \end{subfigure}
      }
  \end{tabular}
 % \centering
\caption{True values, upper bounds, and lower bounds of the two-dimensional estimated states using our strip-based approach in Algorithm~\ref{alg:1}. Every node has four neighbors.}
    \label{fig:set-uplowboundofx}%\vspace{-4mm}
    %  \vspace{-0.2cm}
\end{figure*}
\section{Evaluation} \label{sec:eval}
%,conf:disali2
Our proposed algorithms are implemented in Matlab 2019 on an example similar to the one presented in \citep{conf:cdcdsse,conf:disali2}, where a network of eight nodes attempts to track the position of a rotating object. %All computations run on a single thread of an Intel(R) Core(TM) i7-8750 with 16 GB RAM. 
We made use of CORA \citep{conf:cora1} for zonotope operations along with implementations from \citep{conf:tersekf}. %Our example is quite representative in terms of set-based state estimation since it includes modeling noise and measurement noise.
The state of each node consists of the unknown two-dimensional position of the rotating object. The state matrix in \eqref{eq:sysmodel} is
\begin{align}
 F=\begin{bmatrix}
0.992 & -0.1247 \\
0.1247 & 0.992 
\end{bmatrix},
\end{align}
and the measurement matrix $H^i$ alternates between $[0 1]$ and $[1 0]$ in the sequence of the taken measurements. We run our proposed algorithms in comparison with the one proposed by Garcia et al.~\citep{conf:dis-kalmaninspired} and DKF~\citep{conf:diffusion} on the same generated data set. The related work in set-based methods does not consider sharing the measurements between the neighbors like our approach, and this affects the estimation results but comes at the cost of extra communication. Thus, we will analyze our algorithms with and without sharing the measurements for a fair comparison. Figure~\ref{fig:set-uplowboundofx} shows the true values, upper bound, and lower bound for each dimension of the estimated state using the strip-based approach in Algorithm~\ref{alg:1} while each node is connected to four neighbors. We start with a set ($160 \times 160 m^2$) covering the whole localization area at the initial point (time step $0$), then it becomes smaller due to receiving measurements and performing geometric intersection to correct the estimated state. In addition, we repeat the same experiments using Algorithm~\ref{alg:2} and present the results in Figure~\ref{fig:berger-uplowboundofx}. 

%Our algorithms are implemented with and without the diffusion step to analyze its effect. 

%trim=1200 400 1500 1500,clip,

%%----------------------  luenberger results ----------------%
%%----------------------  luenberger results ----------------%
%%----------------------  luenberger results ----------------%
%%----------------------  luenberger results ----------------%
%%----------------------  luenberger results ----------------%

%\captionsetup[figure*]{labelfont=bf,textfont=normalfont,singlelinecheck=off,justification=raggedleft}

\begin{figure*}[t]
%\vspace{-0.05cm}
    \centering
    \begin{tabular}{ p{0.50\textwidth}  p{0.50\textwidth}}
    \resizebox{0.49\textwidth}{!}{
            \begin{subfigure}[h]{0.49\textwidth}
      \centering
        \includegraphics[scale=0.4]{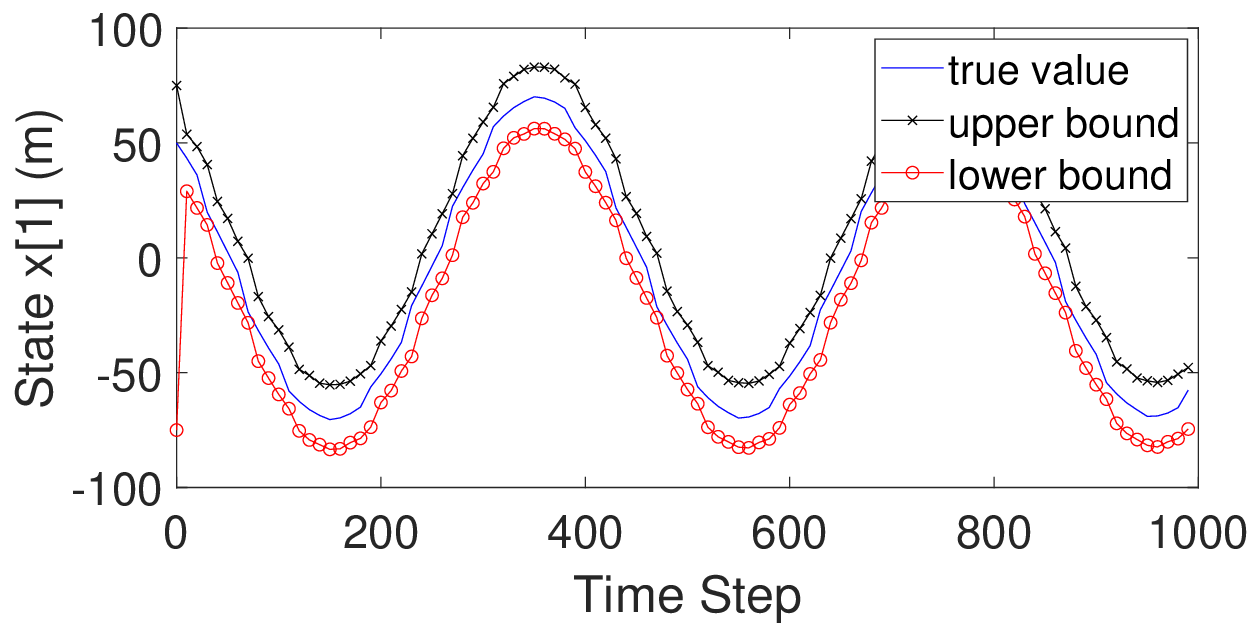}
        \caption{}
        \label{fig:berger-x}
    \end{subfigure}
       } 
   &
   \resizebox{0.49\textwidth}{!}{
            \begin{subfigure}[h]{0.49\textwidth}
      \centering
        \includegraphics[scale=0.4]{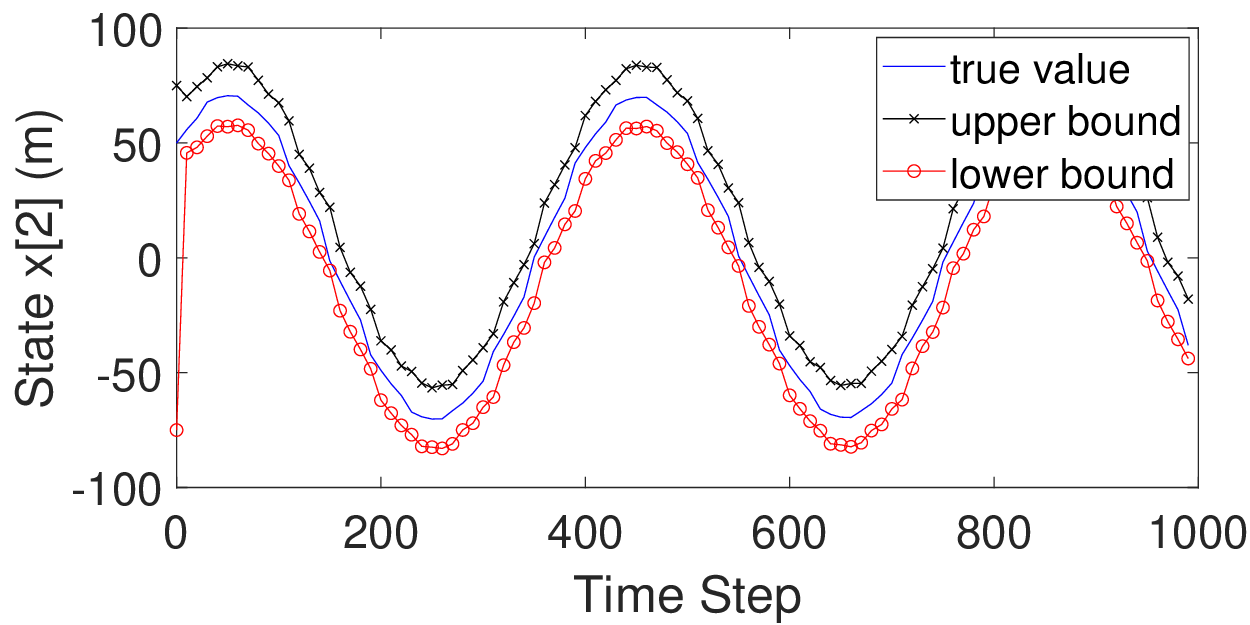}
        \caption{}
        \label{fig:berger-y}
    \end{subfigure}
      }
  \end{tabular}
   %   \captionbox[\raggedleft]{True values, upper bounds and lower bounds of the two-dimensional estimated states using set-propagation approach in Algorithm~\ref{alg:2}, where every  has four neighbors.}
\caption{True values, upper bounds, and lower bounds of the two-dimensional estimated states using our set-propagation approach in Algorithm~\ref{alg:2}. Every node has four neighbors.}
    \label{fig:berger-uplowboundofx}%\vspace{-4mm}
    %  \vspace{-0.2cm}
\end{figure*}

%trim={<left> <lower> <right> <upper>}
\begin{figure*}[t]
%%\vspace{-0.05cm}
    \centering
    \begin{tabular}{ p{0.49\textwidth}  p{0.49\textwidth}}
       \resizebox{0.49\textwidth}{!}{
            \begin{subfigure}[h]{0.49\textwidth}
      \centering
\includegraphics[scale=0.32]{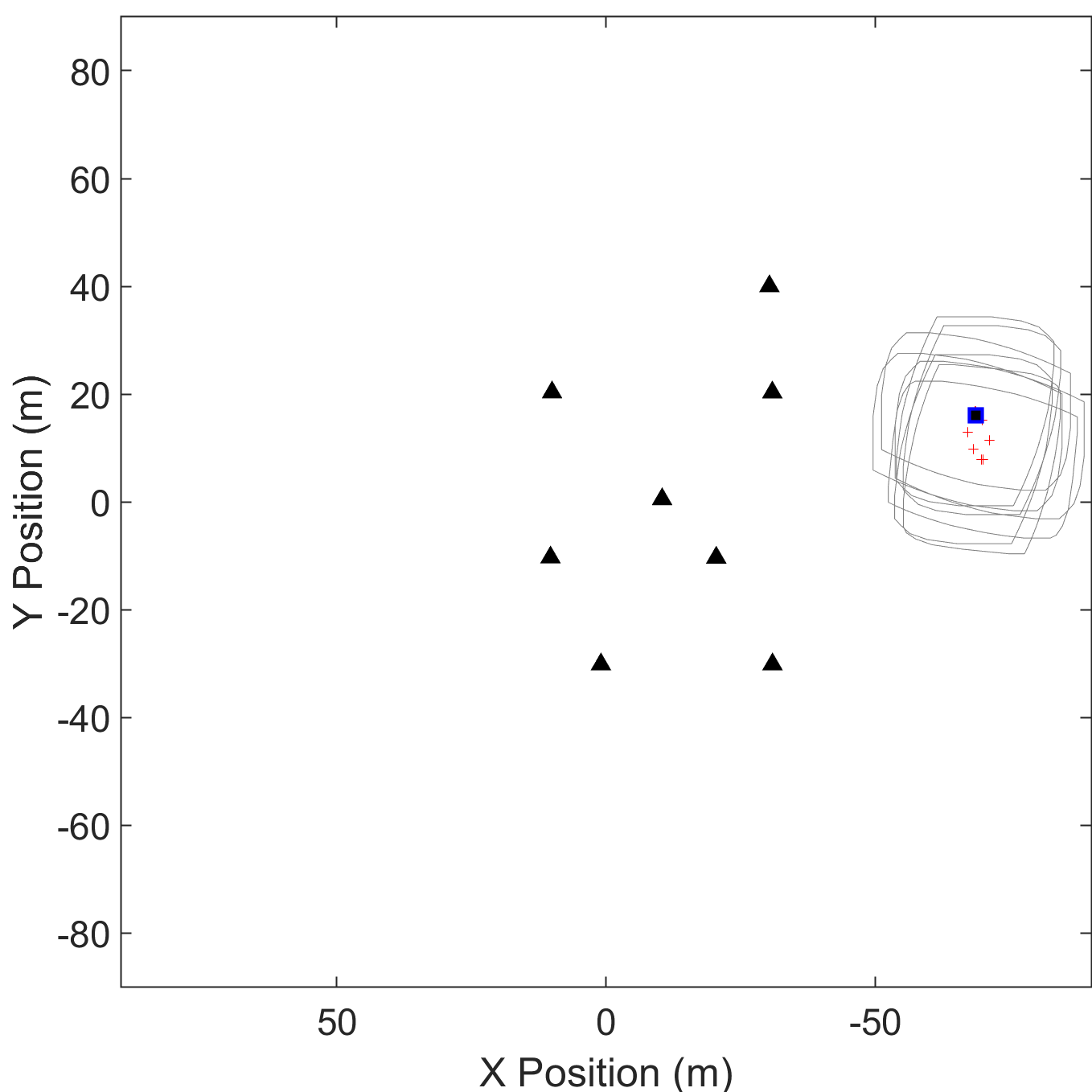}
        \caption{}
        \label{fig:snapwithout}
    \end{subfigure}
      }
      &
        \resizebox{0.49\textwidth}{!}{
            \begin{subfigure}[h]{0.49\textwidth}
      \centering
\includegraphics[scale=0.32]{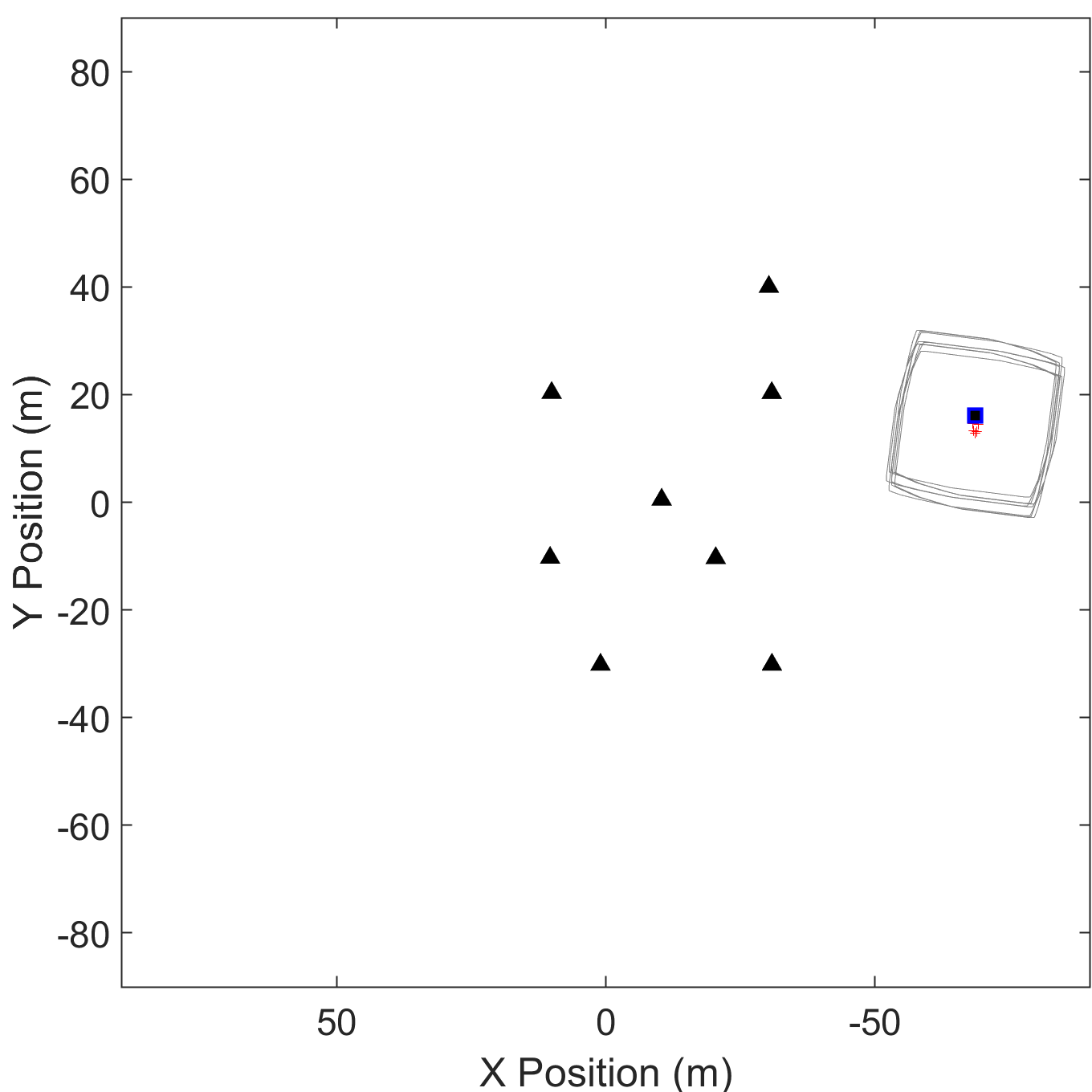}
        \caption{}
        \label{fig:snapwith}
    \end{subfigure}
       } 
  \end{tabular}
\caption{Snapshots of the estimated zonotopes with the red pluses as centers. The triangles are the true positions for the observing nodes. The blue rectangle is the true position of the rotating target. Figure~\ref{fig:snapwithout} shows the distributed estimated zonotopes using Algorithm~\ref{alg:2} without the diffusion step. Figure~\ref{fig:snapwith} shows Algorithm~\ref{alg:2} with the diffusion step.}
    \label{fig:snapstwo}%\vspace{-4mm}
    %  \vspace{-0.2cm}
\end{figure*}

\setlength\tabcolsep{8pt}
\begin{table*}[tbp]
\caption{The mean and standard deviation of the Hausdorff distance (m) between the estimated zonotopes in comparison to \citep{conf:dis-kalmaninspired} and \citep{conf:diffusion}, where every node has two, four, or six neighbors. The results are presented with and without the diffusion step and sharing the measurements over different network connectivities.}
\label{tab:distance}
\centering
\normalsize
\begin{tabular}{l c c c c c c c c}
\toprule
  & && \multicolumn{2}{c}{Six neighbors}&\multicolumn{2}{c}{Four neighbors}&\multicolumn{2}{c}{Two neighbors}\\
  \cmidrule(lr){4-5} \cmidrule(lr){6-7}\cmidrule(lr){8-9}
Algorithm & Meas. sharing &Diffusion & mean & std &mean & std&mean & std\\
\midrule
Alg.~\ref{alg:1} &\cmark &\cmark & 0.242 & 0.160& 0.824 & 0.409 & 2.813 & 2.163\\%with Diffusion 
Alg.~\ref{alg:1} &\xmark &\cmark &3.019 & 2.244&3.899 & 4.062& 5.914 & 3.683\\% without meas
Alg.~\ref{alg:1} &\cmark &\xmark &1.517 & 1.393& 2.703 & 2.118& 3.829 & 2.360\\% without Diffusion
Alg.~\ref{alg:2} &\cmark &\cmark &0.333 & 0.242 &1.897 & 1.332 & 3.871 & 2.362\\%with Diffusion
Alg.~\ref{alg:2} &\xmark &\cmark &1.601 & 1.846& 3.202 & 2.889&5.760 & 3.607\\% without meas
Alg.~\ref{alg:2} &\cmark &\xmark  &1.813 & 1.553 &3.405 & 2.460 & 4.855 & 2.464\\%without Diffusion
Garcia  \citep{conf:dis-kalmaninspired}& - &-& 32.482 & 19.578  & 29.523 & 17.387 & 25.443 & 14.517  \\
DKF  \citep{conf:diffusion}& - & - &0.195 & 0.129& 0.680 & 0.420 &2.811 & 2.099 \\
%\midrule
%Alg.~\ref{alg:1} & \cmark & 0.499 & 0.366  &1.224 & 0.371& 1.961 & 2.495\\%with Diffusion 
%Alg.~\ref{alg:1} & \xmark &2.448 & 1.709 &4.868 & 3.174& 4.385 & 2.350 \\% without Diffusion
%Alg.~\ref{alg:2} & \cmark   & 1.834 & 1.605& 2.115 & 0.638 & 2.899 & 2.018 \\%with Diffusion
%Alg.~\ref{alg:2} & \xmark  &2.620 & 1.838 & 5.013 & 3.252 & 5.590 & 2.613\\%without Diffusion
%\midrule
%Alg.~\ref{alg:1} & \cmark & 1.003 & 0.973&  1.239 & 0.393  & 2.206 & 2.883 \\%with Diffusion 
%Alg.~\ref{alg:1} & \xmark & 2.619 & 2.000 &  5.075 & 3.353& 4.527 & 2.443\\% without Diffusion
%Alg.~\ref{alg:2} & \cmark   & 0.829 & 0.736& 1.915 & 0.746 & 3.131 & 2.223 \\%with Diffusion
%Alg.~\ref{alg:2} & \xmark  & 2.789 & 2.102& 5.181 & 3.378  & 5.685 & 2.755   \\%without Diffusion
%Garcia \citep{conf:dis-kalmaninspired}& - & 33.515 & 13.942 & 31.287 & 13.530 & 27.974 & 12.999\\
\bottomrule
\end{tabular}
\end{table*}

The effect of the diffusion step is analyzed graphically over a network with low connectivity, where every node is connected to two nodes only. Snapshots of the estimated zonotopes by the distributed nodes in Algorithm~\ref{alg:2} are shown in Figure~\ref{fig:snapstwo}. The triangles are the true positions of the monitoring nodes. The estimates are the centers of the zonotopes, which are represented by red pluses. Figure~\ref{fig:snapwithout} and~\ref{fig:snapwith} show the results without and with the diffusion step, respectively. As shown in the aforementioned figures, the diffusion step allows the estimated zonotopes by the distributed nodes to partially consense on a set, which is one of the advantages of adding the diffusion step. 

The Hausdorff distance measures how far two subsets of a metric space are from each other. Thus, as another measure of the estimated zonotope consistency for all the distributed nodes, we calculate the Hausdorff distance between each zonotope at different time steps. We analyze the Hausdorff distance over different network connectivities. The results are reported in Table~\ref{tab:distance} for DKF, Garcia et al. \citep{conf:dis-kalmaninspired}, Algorithms~\ref{alg:1} and~\ref{alg:2} with and without the diffusion step and measurements sharing between the neighbours. The diffusion step enhances the alignment between the estimated zonotopes, significantly affecting a network with low connectivity. For the aforementioned network, every node has access to a lower number of measurements, and thus the diffusion step provides more information to the distributed nodes and enhances the alignment of the estimated zonotopes, estimation results, and radii of the estimated zonotopes. The Hausdorff distance decreases with increasing the number of neighbors in Table~\ref{tab:distance} as the amount of shared information increases with increasing the connectivity. The set in the case of DKF is computed using the $3\sigma$ confidence interval of each node. The Hausdorff distance is smaller in the case of DKF; however, it comes without state containment guarantees. 

%The DKF is close to the Algorithm~\ref{alg:1}

\setlength\tabcolsep{8pt}
\begin{table*}[tbp]
\caption{The mean and standard of the radius (m) of the estimated zonotopes by the proposed algorithms in comparison to \citep{conf:dis-kalmaninspired} and \citep{conf:diffusion}. The results are presented with and without the diffusion step and sharing the measurements over different network connectivities.}
\label{tab:radius}
\centering
\normalsize
\begin{tabular}{l c c c c c c c c}
\toprule
  & && \multicolumn{2}{c}{Six neighbors}&\multicolumn{2}{c}{Four neighbors}&\multicolumn{2}{c}{Two neighbors}\\
  \cmidrule(lr){4-5} \cmidrule(lr){6-7}\cmidrule(lr){8-9}
Algorithm & Meas. sharing &Diffusion & mean & std &mean & std&mean & std\\
\midrule
Alg.~\ref{alg:1} &\cmark&\cmark & 11.877 & 0.057  &12.626 & 0.221 &15.104 & 0.442\\%with Diffusion 
Alg.~\ref{alg:1} &\xmark &\cmark &22.267 & 3.507&24.405 & 2.969& 26.973 & 4.012\\% without meas
Alg.~\ref{alg:1} &\cmark&\xmark &13.312 & 0.419&13.084 & 0.272 & 12.920 & 0.210\\% without Diffusion
Alg.~\ref{alg:2} &\cmark&\cmark &13.266 & 1.235&13.515 & 0.432& 15.257 & 0.393\\%with Diffusion
Alg.~\ref{alg:2} &\xmark &\cmark &20.943 & 3.598&21.699 & 3.142& 21.241 & 3.808\\% without meas
Alg.~\ref{alg:2} &\cmark&\xmark  & 16.690 & 1.443  &17.174 & 0.918& 18.531 & 1.330 \\%without Diffusion
Garcia  \citep{conf:dis-kalmaninspired}& - &-&53.681 & 21.992 & 49.932 & 19.373 & 44.671 & 16.040 \\
DKF  \citep{conf:diffusion}& - &-&3.671 & 0.201&3.949 & 0.283 & 4.463 & 0.482 \\
%\midrule
%Alg.~\ref{alg:1} & \cmark &12.461 & 5.253 &13.467 & 5.175 & 13.261 & 5.188 \\%with Diffusion 
%Alg.~\ref{alg:1} & \xmark &12.242 & 5.281 &13.092 & 5.296 &13.499 & 5.194\\% without Diffusion
%Alg.~\ref{alg:2} & \cmark &13.584 & 5.283& 14.827 & 5.085&  15.715 & 5.012 \\%with Diffusion
%Alg.~\ref{alg:2} & \xmark  & 13.264 & 5.446 & 14.040 & 5.179 & 16.119 & 5.139 \\%without Diffusion
%\midrule
%Alg.~\ref{alg:1}  & \cmark &12.897 & 5.256 &13.431 & 5.180& 13.375 & 5.178 \\%with diffusion 
%Alg.~\ref{alg:1}  & \xmark & 12.242 & 5.281&13.092 & 5.296& 13.499 & 5.194\\ %without diffusion
%Alg.~\ref{alg:2}   & \cmark &13.154 & 5.322& 14.007 & 5.130& 15.652 & 5.013\\ %with diffusion
%Alg.~\ref{alg:2}  & \xmark  &13.264 & 5.446&14.040 & 5.179& 16.119 & 5.139\\%without diffusion
%Garcia \citep{conf:dis-kalmaninspired}& - &51.785 & 18.108&49.270 & 17.291&44.390 & 16.131\\
\bottomrule
\end{tabular}
\end{table*}

\begin{table*}[tbp]
\caption{The mean and standard deviation of the localization error (m) of the center of the estimated zonotopes using the proposed algorithms in comparison to \citep{conf:dis-kalmaninspired} and \citep{conf:diffusion} over different network connectivities, where every node has two, four, or six neighbors. The results are presented with and without the diffusion step and sharing the measurements over different network connectivities.}
\label{tab:loc_err}
\centering
\normalsize
\begin{tabular}{l c c c c c c c c}
\toprule
  & && \multicolumn{2}{c}{Six neighbors}&\multicolumn{2}{c}{Four neighbors}&\multicolumn{2}{c}{Two neighbors}\\
  \cmidrule(lr){4-5} \cmidrule(lr){6-7}\cmidrule(lr){8-9}
Algorithm & Meas. sharing &Diffusion & mean & std &mean & std&mean & std\\
\midrule
Alg.~\ref{alg:1} & \cmark&\cmark &2.227 & 0.490 & 2.419 & 0.485 &3.485 & 0.486\\%with Diffusion 
Alg.~\ref{alg:1} &\xmark &\cmark &3.145 & 4.226&3.020 & 4.095&3.980 & 4.471\\% without meas
Alg.~\ref{alg:1} & \cmark&\xmark &2.463 & 0.409&3.528 & 0.315 &5.605 & 0.477\\% without Diffusion
Alg.~\ref{alg:2} & \cmark&\cmark &1.495 & 0.513& 0.892 & 0.415& 4.882 & 0.871\\%with Diffusion
Alg.~\ref{alg:2} &\xmark &\cmark &3.630 & 4.191&6.475 & 3.512&13.311 & 4.060\\% without meas
Alg.~\ref{alg:2} & \cmark&\xmark  &2.489 & 0.489 &3.550 & 0.459& 5.853 & 0.761 \\%without Diffusion
Garcia  \citep{conf:dis-kalmaninspired}& -&- &11.729 & 1.573 & 11.877 & 1.571& 12.059 & 1.717 \\
DKF  \citep{conf:diffusion}& -&- &2.207 & 0.468&2.403 & 0.464 &3.773 & 0.528 \\
%\midrule
%Alg.~\ref{alg:1} & \cmark &3.228 & 5.870 & 3.262 & 5.850 & 3.627 & 5.748\\%with Diffusion 
%Alg.~\ref{alg:1} & \xmark &3.477 & 5.807 &4.292 & 5.638& 6.009 & 5.582\\% without Diffusion
%Alg.~\ref{alg:2} & \cmark & 2.342 & 6.198 & 4.313 & 6.289 & 13.202 & 6.226 \\%with Diffusion
%Alg.~\ref{alg:2} & \xmark  &3.488 & 5.810 & 4.283 & 5.642 & 6.236 & 5.649 \\%without Diffusion
%\midrule
%Alg.~\ref{alg:1}  & \cmark  &3.220 & 5.870& 3.263 & 5.850 &  3.595 & 5.755 \\%with diffusion 
%Alg.~\ref{alg:1}  & \xmark & 3.477 & 5.807& 4.292 & 5.638& 6.009 & 5.582\\%without diffusion
%Alg.~\ref{alg:2}  & \cmark   & 2.342 & 6.198& 2.645 & 6.364& 4.541 & 6.080\\%with diffusion
%Alg.~\ref{alg:2}  & \cmark   & 2.605 & 6.089&  2.645 & 6.364& 12.835 & 6.156\\ %with diffusion
%Alg.~\ref{alg:2}  & \xmark & 3.488 & 5.810 &4.283 & 5.642& 6.236 & 5.649 \\ %without diffusion 
%Garcia  \citep{conf:dis-kalmaninspired}& - &11.086& 5.870 &11.420 & 5.861&11.580 & 5.902\\
\bottomrule
\end{tabular}
\end{table*}

%The mean and standard deviation without diffusion are $2.062$m and $1.180$m, respectively. On the other hand, the mean and the standard deviation with diffusion are $0.769$m and $0.654$m, respectively. 

%We run the algorithm, distributed strip-based diffusion observer, with a fully connected network and a network where every node has six neighbors out of eight. Figure~\ref{fig:fullyconnected} shows the localization error of choosing lambda as diffusion Kalman filter option and Frobenius norm option. The same experiments are repeated with six neighbors out of eight in Figure~\ref{fig:6neig}. The mean and the standard deviation are summarized in Table~\ref{tab:loc_err}. We repeat the experiments with distributed  Luenberger diffusion zonotopic filter. The fully connected experiment are presented in Table~\ref{tab:berger_loc_err} and Figure~\ref{fig:fullyconnLuen}. Six out of eight connection results are shown in Figure~\ref{fig:berger_6neig}. The mean and the standard deviation are summarized in Table~\ref{tab:berger_loc_err} for the distributed Luenberger diffusion observer.

One important aspect of the performance of the set-based estimation algorithm is reducing the resulting radius of the over-approximating estimated set. Therefore, we analyze the radii of the estimated zonotopes of the proposed algorithms in comparison to the previous work in \citep{conf:dis-kalmaninspired} and \citep{conf:diffusion}. Table~\ref{tab:radius} shows the mean and standard deviation of the radii with and without the diffusion step and measurement sharing. The diffusion step and the measurement sharing decrease the radii due to the proposed intersection criteria. Moreover, our proposed algorithms with and without the diffusion step are much better than the previous work in \citep{conf:dis-kalmaninspired}. We note that the network with higher connectivity has a smaller radius as the intersection with more strips decreases the estimated set. The center of the estimated zonotope is considered a single-point estimate of the proposed algorithms. Therefore, we report the localization error of the estimated centers by the proposed algorithms in Table~\ref{tab:loc_err}. The diffusion significantly enhances the center estimate of the proposed algorithms. Again, the diffusion step is more effective in a network with a low connectivity.

Table~\ref{tab:exectime} shows the execution time of each step in the proposed algorithms while again changing the number of neighbors. To measure the execution time, we run each step $500$ times with randomly generated zonotopes with $20$ generators and take the average execution time. The measurements were taken on an 11$^{th}$ Generation Intel(R) Core(TM) i7-1185G7 processor with 16.0 GB RAM. The time update step does not depend on the number of neighbors. %The diffusion execution time is much faster than the execution time of zonotope intersection using constrained zonotope, which is proposed in \citep{conf:const_zono} and implemented in \citep{conf:cora1}.

\begin{table}[tbp]
\caption{Execution time in $\mu$ seconds of the proposed measurement, diffusion, Luenberger, and time updates in Algorithms~\ref{alg:1} and~\ref{alg:2} with a different number of neighbors.}
\label{tab:exectime}
\centering
\normalsize
\begin{tabular}{l  c c c}
\toprule
   & \multicolumn{3}{c}{Number of neighbors}\\
  \cmidrule(lr){2-4} 
Step  & Six & Four & Two\\
\midrule
Measurement update & $76$   &  $70$ & $63$ \\
Diffusion update & $64$  & $48$ & $31$\\
Time update &  $24$ & $24$ & $24$ \\
Luenberger update & $77$ & $73$ & $64$ \\
%CZ intersection \citep{conf:const_zono} & $8 \times10^6 $&  $3 \times10^6$ &   $7 \times10^5$\\
\bottomrule
\end{tabular}
\end{table}

The main challenges in proposing new set-based observers are mainly choosing the set representation. We chose zonotopes as one can efficiently compute linear maps and Minkowski sums – both are essential operations for set-based observers. In addition, selecting the appropriate optimization function for computing the observer gain, which can maintain low computation costs and high accuracy in comparison to the standard volume minimization technique, was a challenge. We ended up choosing the Frobenius norm as a lightweight indication of the volume of the zonotope.

\section{Conclusion} \label{sec:conc}
We propose two distributed strip-based and set-propagation observers using a diffusion strategy. Our algorithms remove the need for a fusion center. They only require every node to communicate with its neighbors: first, to share the measurements and second, to share the estimates. The diffusion step ensures that information is propagated throughout the network in order to converge to the best estimate and provide consistency between the estimated sets. We propose a new over-approximation for zonotopes intersection to compose the diffusion step and evaluate our algorithms in a localization example of a rotating object.

\section*{Acknowledgements}
We gratefully acknowledge partial financial support by the project justITSELF, funded by the European Research Council (ERC) under grant agreement No 817629, the project interACT under grant agreement No 723395, and the CONCORDIA cyber security project No. 830927; these projects are funded within the EU Horizon 2020 program. The Swedish Research Council and Knut and Alice Wallenberg Foundation also supported this work.

%\begin{figure}[h!]
%\centering
%\includegraphics[scale=0.7]{sys.png}
%\caption{Machine Learning for secure CPS}
%\label{fig:sys}
%\end{figure}

%\bibliographystyle{ACM-Reference-Format}
%\bibliography{references}

%\ifCLASSOPTIONcaptionsoff
%  \newpage
%\fi

%\small
\bibliographystyle{model1-num-names}
\bibliography{references}
%\bibliography{cas-refs}

\end{document}